\documentclass[12pt]{article}
\usepackage{amsmath}
\usepackage{graphicx}
\usepackage{enumerate}
\usepackage{natbib}
\usepackage{amssymb}
\usepackage{amsthm}
\usepackage{xcolor}
\usepackage{CJK}
\usepackage{blkarray}
\usepackage{mathrsfs}
\usepackage{booktabs}
\usepackage{epstopdf}
\usepackage{diagbox}
\usepackage{arydshln}
\usepackage{amsfonts}
\usepackage{multicol}
\usepackage{multirow}
\usepackage{algorithm}
\usepackage{subcaption}
\usepackage{algorithmic}
\usepackage{rotating}
\usepackage{setspace}
\usepackage{graphicx}
\usepackage{xr}
\usepackage[bookmarks=false,
colorlinks,
linkcolor=blue,
anchorcolor=blue,
citecolor=blue,
urlcolor=black
]{hyperref}

\newtheorem{thm}{Theorem}

\newtheorem{rmk}{Remark}

\newcommand{\dT}{\mathsf{T}}

\def\squarebox#1{\hbox to #1{\hfill\vbox to #1{\vfill}}}
\def\boxit#1{\vbox{\hrule\hbox{\vrule\kern6pt
			\vbox{\kern6pt#1\kern6pt}\kern6pt\vrule}\hrule}}

\newcommand{\blind}{1}

\allowdisplaybreaks

\date{}

\begin{document}

	\def\spacingset#1{\renewcommand{\baselinestretch}%
		{#1}\small\normalsize} \spacingset{1}

	
	\if1\blind
	{
		\title{\bf Refining Genetic Discoveries of Group Knockoffs via A Feature-level Filter}
		\author{Jiaqi Gu\\
			Department of Mathematics and Statistics, University of South Florida\\
			and\\
			Zhaomeng Chen\\
			Department of Statistics, Stanford University\\
			and\\
			Zihuai He\\
			Department of Neurology and Neurological Sciences, Stanford University\\
			Department of Medicine (Biomedical InformaticsResearch), Stanford University\\
			Department of Biomedical Data Science, Stanford University
		}
		\maketitle
	} \fi
	
	\if0\blind
	{
		\bigskip
		\bigskip
		\bigskip
		\begin{center}
			{\Large\bf Supplementary Materials of ``Refining Genetic Discoveries of Group Knockoffs via A Feature-level Filter"}
		\end{center}
		\medskip
	} \fi
	
	\bigskip
	\begin{abstract}
	Identifying variants that carry substantial information on the trait of interest remains a core topic in genetic studies. In analyzing the EADB-UKBB dataset to identify genetic variants associated with Alzheimer's disease (AD), however, we recognize that both existing marginal association tests and conditional independence tests using existing knockoff filters suffer either power loss or lack of informativeness, especially when strong correlations exist among variants. To address these limitations, we propose a new feature-versus-group (FVG) filter that {achieves balance between the power and precision} in identifying important features from a set of strongly correlated features using group knockoffs. In extensive simulation studies, the FVG filter controls the expected proportion of false discoveries and identifies important features {in smaller catching sets without large power loss}. Applying the proposed method to the EADB-UKBB dataset, we discover important variants from 89 loci (similar to the most powerful group knockoff filter) with catching sets of substantially smaller size and higher purity {and verify the biological informativeness of our discoveries}.
	\end{abstract}
	
	\noindent%
	\vfill
	
	\newpage
	
	\section{Introduction}\label{Introduction}
	
	\subsection{Conditional Independence Test to Discover Disease-associated Variants}\label{Tests}
	
	\noindent Testing conditional independence between a set of features $\textbf{X}=(X_1,\ldots,X_p)^\dT$ and the response of interest $Y$ is an important topic in various research areas, including causal inference \citep{Peters2015,Cai2022}, genetic analysis \citep{Khera2017,Zhu2018} and graphical model learning \citep{Deka2016,Tugnait2022}. In the era of big data, as the number of features increases, the need for statistical approaches for simultaneous inference of conditional independence between hundreds of thousands of features and the response keeps emerging.

	Specifically, developing statistical methods for conditional independence tests with type-I error rate control has been an important and popular research topic of large-scale genome-wide association studies (GWAS; \citealp{Tang2021,Hou2023,Morra2023}) in recent years. With the ultimate goal of identifying novel targets for the development of genomic-driven medicine, the contribution of more genetic variants to complex phenotypes has been investigated in GWAS. However, our ability to detect causal genetic variants and to translate discoveries into insights of the underlying genetic mechanisms does not increase proportionally to the growth of the scale of GWAS. On one hand, with the stringent criterion to control the family-wise error rate (FWER), conventional GWAS are usually suboptimal in terms of the statistical power of detecting important variants. On the other hand, because conventional GWAS base their inference on marginal association models which regress $Y$ on one $X_j$ at a time, they can only provide blurred results with a lot of proxy variants correlated with the true causal variants \citep{Schaid2018}. Both issues become more severe when more variants are sequenced at a higher resolution because the $p$-value threshold of FWER control also becomes more stringent and there are more proxy variants strongly correlated with true causal variants. This calls for new statistical inference methods with better theoretical properties and empirical performance in multiple testing of conditional independence for real-world genetic data analysis.
	

	In this article, we investigate the stage I meta-analysis of the European Alzheimer \& Dementia Biobank (EADB) dataset and the UK Biobank (UKBB) dataset \citep{Bellenguez2022} to identify genetic variants associated with AD.  This meta-analyzed EADB--UKBB dataset is deposited in the European Bioinformatics Institute GWAS Catalog (\url{https://www.ebi.ac.uk/gwas/}; accession number: GCST90027158), including $39,106$ clinically diagnosed AD cases, $46,828$ proxy cases, and $401,577$ control cases from $15$ European countries. We focus on $650,706$ directly genotyped variants with minor allele frequency (MAF) at least $0.01$, resulting in a processed dataset $\mathbb{G}$ of $487,511$ observations and $650,706$ features. With such a large dataset, we aim to discover as many as possible important genetic variants that carry independent information of AD without making too many false discoveries. Mathematically, such a task can be transferred into a multiple testing problem of conditional independence between a large set of variants $\textbf{X}=(X_1,\ldots,X_p)^\dT$ and the binary response $Y$ which indicates whether a case is clinically diagnosed AD. 
	
	\subsection{Existing Approaches and Their Limitations}\label{Existing}
	
	\noindent Recent decades have witnessed fruitful works in the development of multiple testing procedures. To control the family-wise error rate (FWER) of making at least one false discovery, a lot of $p$-value  based methods have been developed as improvements of the classic Bonferroni correction, including the Šidák correction \citep{Sidak1967}, Holm's step-down procedure \citep{Holm1979} and Hochberg's step-up procedure \citep{Hochberg1988}. However, FWER has been criticized for its conservativeness in many settings, especially when most of the signals are weak. To perform powerful inference, \citet{Benjamini1995} proposed the false discovery rate (the expected proportion of false discoveries in all discoveries) as an alternative type-I error rate measure and developed a procedure to perform multiple testing with provable FDR control under the assumption that $p$-values are independent or positively dependent \citep{Yekutieli2001}. Inspired by this pioneering work, \citet{Yekutieli2001} developed a generalized procedure without assumptions on dependencies among $p$-values, while \citet{Storey2002} and \citet{Whittemore2007} investigated the Bayesian counterpart of FDR. However, as most of the aforementioned methods require $p$-values of nulls for all features, their feasibility on multiple testing of conditional independence becomes questionable when the number of features $p$ greatly exceeds the sample size $n$. The main reason is that $p$-values of conditional independence under high-dimensional scenarios are nontrivial and generally not available except in a very few specific cases \citep{Tibshirani2016,Lee2016}. 

	Unlike all the above methods, the recently developed model-X knockoff filter \citep{Candes2018}  and its variations \citep{Gimenez2019,Gimenez2019b,Barber2019,Bates2020,Huang2020,Ren2023,Li2024} are a family of approaches that can  control finite-sample FDR in multiple testing of conditional independence hypotheses at the feature-versus-feature level
	\begin{equation}
		\label{H_ii}
		H^{(\text{ff})}_{j}:X_j\perp Y|\textbf{X}_{-j},\text{ where  }\textbf{X}_{-j}=(X_1,\ldots,X_{j-1},X_{j+1},\ldots,X_p)^\dT, \quad j=1,\ldots,p,
	\end{equation}
	without any assumption on the conditional distribution $Y|\textbf{X}$. Specifically, these methods first construct knockoffs
	$\widetilde{\textbf{X}}$ that mimic the dependency among features $\textbf{X}$ and are not dependent on the response $Y$. By doing so, knockoffs play the role of synthetic controls: if $H^{(\text{ff})}_{j}$ is true, ${X}_j$ and $\widetilde{X}_j$ are indistinguishable even with information of $Y$; if $H^{(\text{ff})}_{j}$ is false, ${X}_j$ and $\widetilde{X}_j$ possess different dependencies with $Y$. By comparing the dependencies between $Y$ and  $\textbf{X}$ with those between $Y$ and  $\widetilde{\textbf{X}}$, these methods have both controlled FDR and superior power in detecting true causal features as validated in extensive simulation studies. 
	
	However, in the analysis of real-world genetic data {where genetic variants are the features to be selected}, methods that infer $H^{(\text{ff})}_{j}${'s} usually suffer severe power loss. The main reason is that linkage disequilibrium (LD) could introduce high correlations among genetic variants, making it difficult to distinguish variants that are truly associated with the response (whose $H^{(\text{ff})}_{j}$'s are false) from their strongly correlated null variants (whose $H^{(\text{ff})}_{j}$'s are true). This would significantly decrease the number of discoveries as shown by the simulation studies in  \citet{Barber2015} and \citet{Candes2018}. Such a power loss issue would deteriorate when more genetic variants with strong correlations are recorded under a growing resolution of sequencing. 
	For example, when we apply the model-X knockoff filter \citep{Candes2018} to the EADB-UKBB dataset \citep{Bellenguez2022} , only 162 AD-associated genetic variants from 38 loci (here, we define two loci as different if they are at least 1Mb away from each other) are identified under the target FDR level $0.10$ as shown in Table \ref{Tab:motivation}. 
	Compared with the 54 loci discovered by marginal association test ($p$-value threshold: $5\times 10^{-8}$) in \citet{Bellenguez2022}, although model-X knockoff filter identified 7 new loci, it misses several important AD loci with strong signals, including the ``ZCWPW1"  locus on chromosome 7, the ``PTK2B"  locus on chromosome 8, the ``ECHDC3"  locus  on chromosome 10, the "MS4A4E" locus on chromosome 11, the ``FERMT2"  locus on chromosome 14, the ``FAM157C"  locus on chromosome 16 and the ``SIGLEC11"  locus on chromosome 19. This can seen in Figures \ref{fig:motivating_cluster} (a)-(b) where we displays the difference of inference results between marginal association test and model-X knockoff filter within a high LD region between positions 31571218 $\sim$ 32682663,  chromosome 6. As correlation among genetic variants is high in this region, most of genetic variants possess a small $p$-values, which not only implies the existence of contributing variants but also obstacles our distinguishing of the contributing ones from all proxy ones. Such great correlations also result in no discoveries (under FDR $0.10$) of contributing variants of the model-X knockoff filter as the importance score of the most contributing variant does not exceed the data-driven threshold.
	
	\renewcommand{\arraystretch}{0.9}
	\begin{table}
		\centering\caption{Summary of results by applying different existing methods to the EADB-UKBB dataset, where bold values indicate the best performance under different evaluation metrics.}\label{Tab:motivation}
		\resizebox{\columnwidth}{!}{\begin{tabular}{lrrrr}
				\toprule
				\multirow{3}*{Method}&Number of &Average number of &Average&Average \\
				&identified&variants per & size of& purity$^\ast$ of\\
				&loci&identified locus&catching sets$^\star$&catching sets\\
				\midrule
				Marginal association test&54&17.093&17.093&0.489\\
				Model-X knockoff filter&38&\textbf{4.263}&1.000$^\dagger$&1.000$^\dagger$\\
				Group knockoff filter&\textbf{91}&19.923&9.157&0.651\\
				KnockoffZoom (2 layers)&\textbf{91}&15.152&5.621&0.851\\
				KnockoffZoom (3 layers)&\textbf{91}&14.681&5.409&0.858\\
				KeLP (2 layers)&45&10.422&2.535&0.918\\
				KeLP (3 layers)&38&10.605&\textbf{2.488}&\textbf{0.932}\\
				\bottomrule
				\multicolumn{5}{l}{\small $^\star$: Catching sets refer to }\\
				\multicolumn{5}{l}{\small (a) sets of variants within different loci identified by marginal association test;}\\
				\multicolumn{5}{l}{\small (b) variants identified by the model-X knockoff filter;}\\
				\multicolumn{5}{l}{\small (c) variant groups identified by the group knockoff filter.}\\
				\multicolumn{5}{l}{\small (d) resolution-adaptive discoveries made by KnockoffZoom and KeLP that can be either variants identified}\\
				\multicolumn{5}{l}{\small  \hspace{0.35cm} at the variant level or variant groups identified at the group level containing no identified variants.}\\
				\multicolumn{5}{l}{\small $^\ast$: Purity refers to the minimum absolute correlation within the catching set.}\\
				\multicolumn{5}{l}{\small $^\dagger$: Size and purity of catching sets provided by model-X knockoff filter are trivially 1 and thus not included}\\
				\multicolumn{5}{l}{\small \hspace{0.2cm} in comparison.}\\
		\end{tabular}}
	\end{table}

	\begin{figure}[t]
		\centering
		\begin{minipage}{\linewidth}
			\centering
			\includegraphics[width=0.39\linewidth]{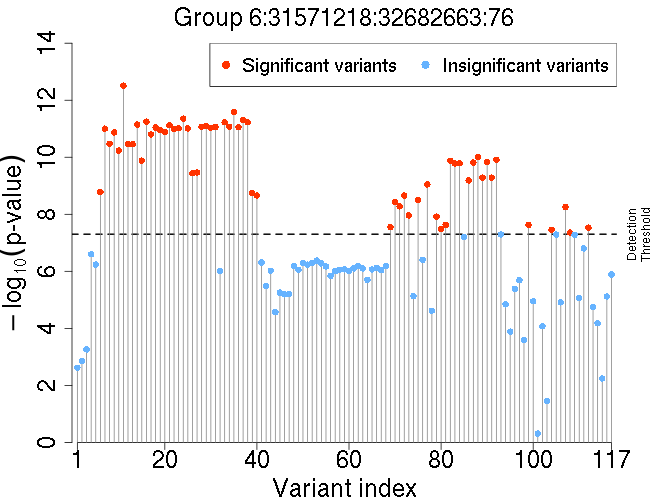}
			\includegraphics[width=0.5\linewidth]{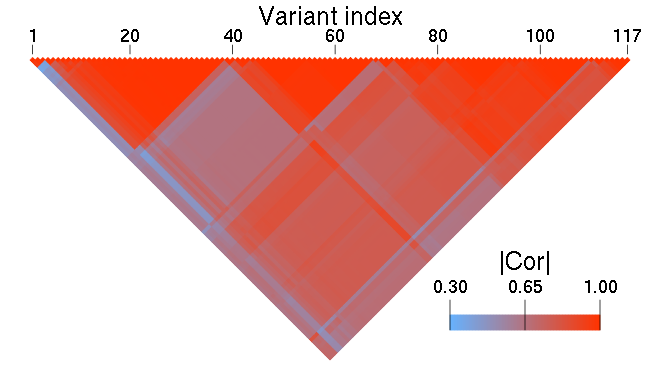}
			\par{(a) Manhattan plot of marginal association test and the correlation plot within the high-LD region.}
		\end{minipage}
		\begin{minipage}[t]{0.49\linewidth}
			\centering
			\includegraphics[width=0.78\linewidth]{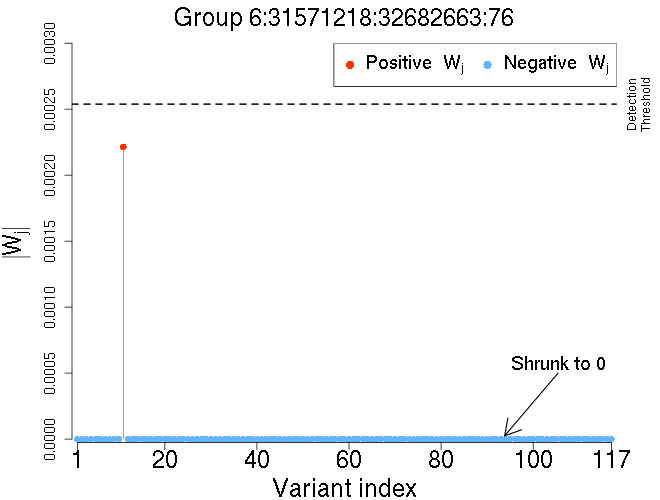}
			\par{{(b) Manhattan plot of feature importance statistics of \\different genetic variants under model-X knockoff filter.}}
		\end{minipage}
		\begin{minipage}[t]{0.49\linewidth}
			\centering
			\includegraphics[width=0.78\linewidth]{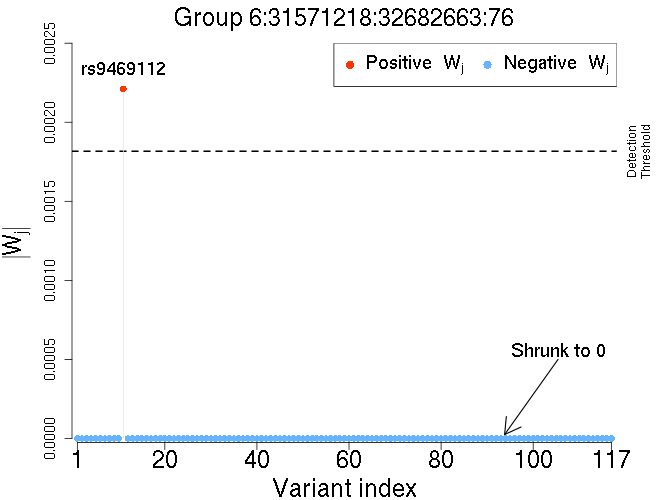}
			\par{{(c) Manhattan plot of the contribution of different variants to the importance statistics of the group under the group knockoff filter.}}
		\end{minipage}
		\caption{Comparison of inference results provided by marginal association test \citep{Bellenguez2022}, model-X knockoff filter \citep{Candes2018} and the group knockoff filter \citep{Dai2016} in a high-LD region between positions 31571218 $\sim$ 32682663,  chromosome 6. (Displaying order of genetic variants is reorganized according to the hierarchical clustering of $(1-|\text{cor}(G_i,G_{j})|)$ for better visualization.)}
		\label{fig:motivating_cluster}
	\end{figure} 
	
	To address this power loss issue in real-world genetic data analysis, people turn their inference target from the feature-versus-feature level to the group-versus-group level and use group knockoff filter to detect signals lie in high LD region. First introduced by \citet{Dai2016} and developed by multiple works \citep{Katsevich2019,Sesia2020,Sesia2021,Spector2022,Chu2024,Gablenz2024}, the group knockoff filter defines groups $B_1,\ldots,B_k$ such that {features with strong correlations are empirically allocated in the same group} and perform multiple testing on \begin{equation}
		\label{H_gg}
		H^{\text{(gg)}}_{k}:\textbf{X}_{B_k}\perp Y|\textbf{X}_{-B_k},\quad k=1,\ldots,K,
	\end{equation}
	at the group-versus-group level, where important features and their strongly correlated null features can be rejected as a group. Because the grouping is done to ensure features from different groups are weakly correlated, the group knockoff filter can circumvent the obstacle of high correlation in inferring $H^{(\text{ff})}_{j}$ with significantly higher power. To apply the group knockoff filter to the EADB-UKBB dataset \citep{Bellenguez2022}, we generate the second-order group knockoffs as follows by assuming that the distribution of $\textbf{X}$ can be well approximated by a multivariate Gaussian distribution. {To empirically construct the genetic variant groups according to the correlation, we first use the UK Biobank directly genotyped data (which also records genetic information of Europeans) as the reference panel and compute correlation $\text{cor}(G_i,G_j)$ between any pair of variants $G_i$ and $G_j$ via the Pan-UKB consortium (\url{https://pan.ukbb.broadinstitute.org}). We then perform the hierarchical clustering (average linkage with cutoff value $0.5$) on the distance matrix {\rm$(1-|\text{cor}(G_i,G_{j})|)_{p\times p}$} and construct $321,569$ variant groups. By doing so, the average absolute correlation of variants (features) from different groups is lower than 0.5.} With second-order group knockoffs generated under such a group structure and the maximum-entropy construction using Algorithm 2 of \citet{Chu2024}, the group knockoff filter identifies 198 AD-associated variant groups from 91 loci under the target FDR level $0.10$. Specifically, among these 198 identified groups, 121 are located outside the strongest APOE/APOC region (chromosome 19, positions 44909011 $\sim$ 45912650), including groups in loci ``ZCWPW1", ``PTK2B", ``ECHDC3", "MS4A4E", ``FERMT2", ``FAM157C"  and``SIGLEC11" that are missed by the model-X knockoff filter. In addition, the group knockoff filter also identifies {43} new loci that are not discovered by \citet{Bellenguez2022}.
	
	However, doing so can only provide a series of AD-associated variant groups without information on which variants contribute more to AD variation within each group, leaving the target of identifying AD-associated variants not fully addressed. This can be observed from Table \ref{Tab:motivation}, where each identified groups given by the group knockoff filter are treated as a catching set of genetic signals. With an average size $19.923$ of catching sets, 30 catching sets contain at least 10 genetic variants and the 5 largest sets contain 314, 232, 117, 84 and 57 variants respectively. With such large sizes, which genetic variants in these catching sets contribute more to AD variation remains unanswered. This can also be seen in Figure \ref{fig:motivating_cluster} (c). Within the high-LD region between positions 31571218 $\sim$ 32682663,  chromosome 6, all genetic variants are clustered in the same group and such a group is identified as a sole catching set. However, {such an inference result} loses the information that all of the importance of this variant group is contributed by the variant \texttt{rs9469112} as shown in Figure \ref{fig:motivating_cluster} (c).
	
	To improve the precision of catching sets of genetic signals, several approaches have been developed to return catching sets of multiple resolutions or in a resolution-adaptive manner. One of the pioneering method is the multilayer knockoff filter \citep{Katsevich2019} which simultaneously infers $H^{(\text{ff})}_{j}$'s and $H^{(\text{gg})}_{k}$'s, while the later developed KnockoffZoom\footnote{In this article, we adopt Algorithm 5 in the supplementary note 4 of \citet{Sesia2020}.} \citep{Sesia2020} and knockoff e-value linear program (KeLP; \citealp{Gablenz2024}) are implemented to identify genetic signals in resolution-adaptive manners by using catching sets at different resolutions. Specifically, KnockoffZoom gradually refines catching sets by sequentially performing group knockoff inference from lower resolutions to higher resolutions and finally using model-X knockoff filter at the feature-versus-feature level. KeLP, in contrast, utilizes the method of \citet{Ren2024} to compute $e$-values of all variant groups at different resolutions and all variants and identifies signals at the finest possible resolution by solving an optimization problem with nesting constraints.

	However, when we apply KnockoffZoom and KeLP to identify contributing genetic variants under the target FDR level $0.10$  {at both the feature-versus-feature level as in the model-X knockoff filter and group-versus-group level as in the group knockoff filter} (which we label as ``(2 layers)"), existing issues are not fully addressed. On one hand, as the inference of KnockoffZoom (2 layers) at the group-versus-group level is the same as the group knockoff filter, KnockoffZoom (2 layers) identifies the same loci as the group knockoff filter and does not {suffer from any power loss}. However, as signal-to-noise ratios of $H^{(\text{ff})}_{j}$'s are low due to high LD, half of groups (99/198) contain no identified variants at the feature-versus-feature level. Thus, KnockoffZoom (2 layers) only records moderate improvement of average size (decreases from 9.157 to 5.621) and purity (increases from 0.651 to 0.851) of catching sets. This can be seen from Figure \ref{fig:motivating_cluster2} (a) in Appendix \ref{Supp_fig}, where importance score of the most contributing variant in the high-LD region between positions 31571218 $\sim$ 32682663,  chromosome 6, remains lower than the data-driven threshold. Even when we add an intermediate layer (obtained by the hierarchical clustering with average linkage and cutoff value $0.25$) between the current two levels in applying KnockoffZoom (we refer to as ``KnockoffZoom (3 layers)"), such a phenomenon retains
	as shown in Figures \ref{fig:motivating_cluster2} (b)-(c). KeLP, on the other hand, {suffer from greater power loss} even compared to the model-X knockoff filter, no matter with 2 layers or 3 layers. The main reason is that computing e-values for resolution-adaptive inference under the target FDR level $0.10$ requires inferences under a target FDR level lower than $0.10$ throughout all resolutions. As only variants and groups with nonzero e-values can be selected by KeLP, lower target FDR levels would decrease the number of variants and groups with nonzero e-values and thus the number of catching sets.
	
	{Another category of commonly used approaches are various Bayesian fine-mapping approaches \citep{Wilson2010,Guan2011,Hormozdiari2014,Chen2015,Benner2016,Wang2020,Zou2022,Li2024a}. For comprehensive review, please refer to \citet{Schaid2018} and \citet{Li2025}. Generally, existing Bayesian fine-mapping approaches fit a parametric model between the response and features (genetic variants) in identified important regions, compute the posterior inclusion probability (PIP) of each feature to be contributive via Bayesian approaches and construct credible sets of features to contain at least one signal with large posterior probability. Potentially, one can perform fine-mapping within variant groups identified by the group knockoff filter and obtain credible sets of smaller size and higher precision. However, using such a two-step heuristic approach still faces the tradeoff between power and precision via the selection of coverage probability threshold.}
	
	%
	%
	%

	\subsection{Our Contribution}\label{Contribution}
	
	In this article, to address the limitations of existing knockoff methods in discovering important variants of AD in the analysis of the EADB-UKBB dataset, we propose a new knockoff filter that {achieves balance between the statistical power and precision} in identifying important features from a set of strongly correlated features. Built upon group knockoffs, {our filter simultaneously selects groups containing important features and performs fine-mapping to identify the most promising features within each selected group. Thus, our method can return refined catching sets that are small in size and high in purity without large power loss.} Specifically, our approach of refinement is motivated by penalized regressions that can efficiently learn the sparse nature of causality via the sparse feature importance, which is leveraged in searching of catching sets under alternative variant-level hypotheses.

	The rest of this article is organized as follows. In Section \ref{Methodology}, we develop our new knockoff filter and perform theoretical analysis. 
	Section \ref{Simulation} illustrates the empirical performance of our method and its advantage over existing alternatives in power and {precision} via simulated experiments. Finally, we apply our method to the AD analysis of the EADB-UKBB dataset in Section \ref{Real} and demonstrate that the method exhibits balance between power and precision. Section \ref{Discussions} concludes with discussions.	
	
	\section{Methodology}\label{Methodology}
	
	\subsection{Problem Statement}\label{problem}
	
	Consider independent and identically distributed (i.i.d.) samples $\{(\textbf{x}_i,y_i)|i=1,\ldots,n\}$ from a joint distribution $f(\mathbf{ X},Y )$, where features $\textbf{X}=(X_1,\ldots,X_p)^\dT$ are partitioned into $K$ disjoint groups $B_1,\ldots,B_K$ {for inference. For example, in analysis of genetic data, features (variants) with high correlation are usually allocated to the same group such that features from different groups are weakly or moderately correlated.} Let $\mathbb{X}=[\textbf{x}_1,\ldots,\textbf{x}_n]^\dT$ denote the data matrix and $\textbf{y}=(y_1,\ldots,y_n)^\dT$ be the vector of responses. Our target is to perform multiple testing on conditional independence hypotheses,
	\begin{equation}\label{H_ig}
		H_{j}^{(\text{fg})}:X_{j}\perp Y|\textbf{X}_{-B_k},\quad k=1,\ldots,K,\quad j\in B_k,
	\end{equation}
	where $\textbf{X}_{-B_k}=(\textbf{X}^\dT_{B_1},\ldots,\textbf{X}^\dT_{B_{k-1}},\textbf{X}^\dT_{B_{k+1}},\ldots,\textbf{X}^\dT_{B_K})^\dT$ is the subvector of $\textbf{X}$ created by excluding all features in the same group $B_k$ of $X_j$. Here, we refer $H_{j}^{(\text{fg})}$ as the conditional independence hypothesis at the feature-versus-group level as it depicts the independence between a feature and the response conditional on other feature groups which omits information of other features in the same group.  Under the belief that the response $Y$ only depends on a relatively small number of features as validated in many GWAS studies, we would like to find as many important features as possible without making too many false discoveries with respect to $H_{j}^{(\text{fg})}$'s. In other words, our target is to obtain a rejection set $\mathcal{R}^{(\text{fg})}=\{j|H_{j}^{(\text{fg})}\text{ is rejected}\}$ such that the false discovery rate (FDR, the expectation of false discovery proportion),
	\begin{equation}
		\label{FDR_H_ig}
		\text{FDR}^{(\text{fg})}=\mathbf{E}\{\text{FDP}^{(\text{fg})}\},\quad \text{where }\text{FDP}^{(\text{fg})}=\frac{\#(\mathcal{R}^{(\text{fg})}\cap \mathcal{H}_0^{(\text{fg})})}{1\vee\#\mathcal{R}^{(\text{fg})}},\quad \mathcal{H}_0^{(\text{fg})}=\{j|H_{j}^{(\text{fg})}\text{ is true}\},
	\end{equation}
	is controlled under the target level $\alpha$ ($0<\alpha<1$).
	
	Our method is under the same setup of \citet{Candes2018} that the distribution of features $\textbf{X}$ is known completely or well approximated up to the first two moments, while no information is provided regarding the conditional distribution $Y|\textbf{X}$. {That is to say, our framework does not assume that the contribution of $X_j$'s to $Y$ are additive (independent) and the FDR control remains valid even the underlying conditional distribution $Y|\textbf{X}$ is deviated from the linear model.} {Our method differs from existing approaches, which test either \eqref{H_ii} at the feature-versus-feature level or \eqref{H_gg} at the group-versus-group level, as follows. 
		\begin{itemize}
			\item Compared to those that infer $H^{(\text{ff})}_{j}$'s, our method can utilize group knockoffs that are more liberal in construction to achieve higher power in identifying important features. 
			\item It is true that in most of GWAS studies where highly correlated features are allocated in the same groups, our hypotheses $H_{j}^{(\text{fg})}$'s are equivalent to $H_{k}^{(\text{gg})}$'s at the group-versus-group level for $j\in B_k$. That is to say, including feature $X_j$ in $\mathcal{R}^{(\text{fg})}$ implies that the group $B_k$ contains important features. The difference is that our method simultaneously selects groups containing important features and performs fine-mapping to identify the most promising features within each selected group. Thus, instead of returning variant groups with many variants as large and impure catching sets, our method can further refine the selected variant groups. In addition, our method can achieve better balance between power and precision compared to the two-step heuristic alternative -- applying post-selection fine-mapping to groups selected by the group knockoff filter.
		\end{itemize}
	}
	
	However, such a gain does not come free as new challenge is posted that no existing filters can provide control on $\text{FDR}^{(\text{fg})}$, which is elaborated and addressed in the following sections. 
	
	\subsection{Feature Importance Scores Under Group Knockoffs}\label{score}
	
	To infer conditional independence hypotheses \eqref{H_ig} with group knockoffs, our inference procedure, similar to the existing ones, include four main steps as follows.
	
	\begin{itemize}
		\item[$\diamond$] \textbf{(Knockoffs Construction):} Construct the knockoff copy
		$\widetilde{\textbf{X}}$ of $\textbf{X}$ such that $\widetilde{\textbf{X}}\perp Y|\textbf{X}$ and the joint distribution of $(\textbf{X},\widetilde{\textbf{X}})$ is exchangeable at the group level,
		\begin{equation}\label{group_exchangeability}
			\begin{aligned}[b]
				(\textbf{X},\widetilde{\textbf{X}})_{\text{swap}(\mathcal{S})}{\displaystyle\mathop{=\joinrel=}^{D}}(\textbf{X},\widetilde{\textbf{X}}),\quad \forall \mathcal{S}\subset\{1,\ldots,K\},
			\end{aligned}
		\end{equation}
		where {\rm$(\textbf{X},\widetilde{\textbf{X}})_{\text{swap}(\mathcal{S})}$} is obtained by swapping {\rm$\textbf{X}_{B_k}$} and {\rm$\widetilde{\textbf{X}}_{B_k}$} in {\rm$(\textbf{X},\widetilde{\textbf{X}})$} for all {\rm$k\in \mathcal{S}$}.
		\item[$\diamond$] \textbf{(Feature Importance Scores Calculation):} For each $H^{(\text{fg})}_{j}$, compute feature importance scores $T_j$ for feature $X_j$ and $\widetilde{T}_j$ for its knockoff copy $\widetilde{X}_j$.
		\item[$\diamond$] \textbf{(Feature Statistics Calculation):} For each $H^{(\text{fg})}_{j}$, summarize importance scores $T_j$ and $\widetilde{T}_j$ as a feature statistic $W_j=w_j(T_j,\widetilde{T}_j)$ using an anti-symmetric function $w_j$ such that  $w_j(T_j,\widetilde{T}_j)=-w_j(\widetilde{T}_j,T_j)$.
		\item[$\diamond$] \textbf{(Feature Filtering):} Reject $H^{(\text{fg})}_{j}$ if $W_j$ is positive and large enough (i.e., $T_j$ dominates $\widetilde{T}_j$).
	\end{itemize}
	
	However, directly inputting $W_1,\ldots,W_p$ to the model-X knockoff filter at feature level \citep{Candes2018} cannot provide control on $\text{FDR}^{(\text{fg})}$ of the rejection set. The main reason is that under the group knockoffs model $f(\textbf{X},\widetilde{\textbf{X}})$, $W_1,\ldots,W_p$ do not possess the coin-flipping property (i.e., conditional on $|W_j|$'s and $\text{sign}(W_j)$'s for those false $H^{(\text{fg})}_{j}$'s, $\text{sign}(W_j)$'s for those true $H^{(\text{fg})}_{j}$'s are independently and uniformly distributed on $\{\pm1\}$), the key property that guarantees FDR control in \citet{Candes2018}. Specifically, for most of the widely-used feature importance scores (e.g., the marginal correlations, the Lasso coefficient difference statistics and the Lasso signed max statistics; \citealp{Barber2015}, \citealp{Candes2018}), $W_j$'s for those true $H^{(\text{fg})}_{j}$'s are dependent coin flips. This comes from \eqref{group_exchangeability} where exchangeability only exists among groups, making signs of $W_{j}$'s from the same group not independent.
	
	Thus, we relax the unachievable coin-flipping property to the between-group coin-flipping property and propose feature statistics $W_{j}$'s using anti-symmetric functions $w_j$'s  and
	\begin{equation}\label{ig_type_two}\begin{aligned}
			&\begin{cases}
				T_j=t([\mathbb{X}_j,\mathbb{X}_{B_k\setminus\{j\}},\widetilde{\mathbb{X}}_j,\widetilde{\mathbb{X}}_{B_k\setminus\{j\}},\{\mathbb{X}_{-B_k},\widetilde{\mathbb{X}}_{-B_k}\}],\textbf{y}),\\
				\widetilde{T}_j=t([\widetilde{\mathbb{X}}_j,\widetilde{\mathbb{X}}_{B_k\setminus\{j\}},{\mathbb{X}}_j,{\mathbb{X}}_{B_k\setminus\{j\}},\{\mathbb{X}_{-B_k},\widetilde{\mathbb{X}}_{-B_k}\}],\textbf{y}),
			\end{cases}\quad k=1,\ldots,K\text{, }j\in B_k,\\
			&\text{where }\{\mathbb{X}_{-B_k},\widetilde{\mathbb{X}}_{-B_k}\}\text{ denotes the set of unordered pairs }\{\mathbb{X}_{j'},\widetilde{\mathbb{X}}_{j'}\}\text{ for all }j'\notin B_k.
	\end{aligned}\end{equation}
	{Here, the swapping between $T_j$ and $\widetilde{T}_j$ requires the swapping between $\mathbb{X}_{B_k}$ and $\widetilde{\mathbb{X}}_{B_k}$, while it only requires the swapping between $\mathbb{X}_{j}$ and $\widetilde{\mathbb{X}}_{j}$ in the model-X knockoff filter. However, most of commonly used feature importance scores in the model-X knockoff filter also satisfy (\ref{ig_type_two}). Examples include:} 
	\begin{itemize}
		\item[$\diamond$] \textbf{marginal correlation with response:} absolute values of the empirical marginal correlations
		$T_j=\Big|\widehat{\rho}(X_j,Y)\Big|$ and $\widetilde{T}_j=\Big|\widehat{\rho}(\widetilde{X}_j,Y)\Big|$; 
		\item[$\diamond$] \textbf{joint lasso:} absolute values of lasso estimators $T_j=\Big|\widehat{\beta}_j\Big|$ and $\widetilde{T}_j=\Big|\widehat{\widetilde{\beta}}_j\Big|$ of the linear model \citep{Tibshirani1996}
		\begin{equation}\label{Joint_Lasso}Y=\sum_{j=1}^p\left({\beta}_{j}{{X}}_{j}+\widetilde{\beta}_{j}\widetilde{{X}}_{j}\right)+e;
		\end{equation}
		\item[$\diamond$] \textbf{marginal correlation with lasso residual:} absolute values of estimated marginal correlations $T_j=\Big|\widehat{\rho}(X_j,\hat{e})\Big|$ and $\widetilde{T}_j=\Big|\widehat{\rho}(\widetilde{X}_j,\hat{e})\Big|$ between  $(X_j,\widetilde{X}_j)^\dT$ and the lasso residual of the linear model 
		$$
		Y=\sum_{j^\dagger\notin B_k}\left({\beta}_{j^\dagger}{{X}}_{j^\dagger}+\widetilde{\beta}_{j^\dagger}\widetilde{{X}}_{j^\dagger}\right)+e,\quad \text{where }j\in B_k;
		$$
		\item[$\diamond$] \textbf{separate lasso:} absolute values of lasso estimators $T_j=\Big|\widehat{\beta}_j\Big|$ and $\widetilde{T}_j=\Big|\widehat{\widetilde{\beta}}_j\Big|$ of the separate linear model 
		$$Y={\beta}_{j}{{X}}_{j}+\widetilde{\beta}_{j}\widetilde{{X}}_{j}+\sum_{j^\dagger\notin B_k}\left({\beta}_{j^\dagger}{{X}}_{j^\dagger}+\widetilde{\beta}_{j^\dagger}\widetilde{{X}}_{j^\dagger}\right)+e,\quad \text{where }j\in B_k.
		$$
	\end{itemize}
	In our practice, we suggest using joint lasso importance score because it is computationally efficient and induces sparsity for better refinement of catching sets as stated in Section \ref{problem}.

	As shown in Theorem \ref{thm:uniform_two}, feature importance scores in the form of (\ref{ig_type_two}) satisfy the between-group coin-flipping property.
	\begin{thm}\label{thm:uniform_two}
		For any feature importance scores $T_j$ and $\widetilde{T}_j$ in the form of \eqref{ig_type_two}, feature statistics $W_j=w_j(T_j,\widetilde{T}_j)$ ($j=1,\ldots,p$) with antisymmetric functions $w_j$'s satisfy the between-group coin-flipping property that
		\begin{enumerate}
			\item[$\star$] conditional on $|W_1|,\ldots,|W_p|$ and $\text{\rm sign}(W_{j})$'s for those false $H^{(\text{fg})}_{j}$'s,
			\begin{enumerate}
				\item[$\diamond$] {\rm \textbf{(Uniformity)}} $\text{\rm sign}(W_{j})$ is uniformly distributed on $\{\pm1\}$  for all {\rm$j\in\mathcal{H}_0^{(\text{fg})}$};
				\item[$\diamond$] {\rm \textbf{(Between-group independence)}}
				for any $k\neq k^\dagger$, {\rm$\text{\rm sign}(W_{j})$} and {\rm$\text{\rm sign}(W_{j^\dagger})$} are independent for any $j\in B_k\cap \mathcal{H}_0^{(\text{fg})}$ and $j^\dagger\in B_{k^\dagger}\cap \mathcal{H}_0^{(\text{fg})}$,
			\end{enumerate}
		\end{enumerate}
		if {\rm$H_j^{(\text{fg})}$} implies {\rm$H_k^{\text{(gg)}}$} for all $j\in B_k$ and $k=1,\ldots,K$.
	\end{thm}
	The proof of Theorem \ref{thm:uniform_two} is provided in Appendix \ref{pr:uniform_two}. Specifically, Theorem \ref{thm:uniform_two} relies on a condition ``{\rm$H_j^{(\text{fg})}$} implies {\rm$H_k^{\text{(gg)}}$} for all $j\in B_k$ and $k=1,\ldots,K$". {That is to say, if {\rm$H_k^{\text{(gg)}}$}  is false and group $B_k$ is important, all {\rm$H_j^{(\text{fg})}$}'s ($j\in B_k$) are false and any $j\in B_k$ included in the selection sets would not be counted as false discovery.} Such a condition is generally valid in genetic data analysis, analogous to the usual expectation in fine-mapping that the association between an important variant $X_j$ and the response $Y$ is in the same direction as the direct effect of $X_j$ on $Y$.

	\subsection{Naive Feature Filter under Group Knockoffs}\label{filter}
	
	By Theorem \ref{thm:uniform_two}, feature statistics $W_{j}$'s only satisfy the between-group coin-flipping property and thus directly applying the model-X knockoff filter \citep{Candes2018} would violate the FDR control. However, we can still utilize the between-group independence among $W_{j}$'s for inference with FDR control as follows.

	Recognizing that independence remains true among $\text{sign}(W_j)$'s for those true $H^{(\text{fg})}_{j}$'s from different groups, we consider aligning $W_j$'s in Table \ref{tab:align} such that $W_j$'s for all features in the group $B_k$ are in the $k$-th column and $|W_{(k1)}|\geq |W_{(k2)}|\geq \cdots$. Let $\mathcal{C}_l$ be the set of features whose feature statistics are aligned in the $l$-th row and consider the rejection set in the form of $\mathcal{R}^{(\text{fg})}(t)=\{j|W_j\geq t\}$.
	
	\renewcommand{\arraystretch}{1}
	\begin{table}[h]
		\centering
		\caption{Alignment of feature statistics $W_{j}$'s such that different groups correspond to different columns. Here, $W_{(kl)}$ corresponds to the feature from the $k$-th group that is aligned in the $l$-th row.}
		\label{tab:align}
		\begin{tabular}{ccccc}
			\toprule
			Row&$B_1$&$B_2$&$B_3$&$\cdots$\\
			\midrule
			1&$W_{(11)}$&$W_{(21)}$&$W_{(31)}$&$\cdots$\\
			\hdashline
			2&$W_{(12)}$&$W_{(22)}$&$W_{(32)}$&$\cdots$\\
			\hdashline
			3&$W_{(13)}$&$W_{(23)}$&$W_{(33)}$&$\cdots$\\
			\hdashline
			4&$W_{(14)}$&$W_{(24)}$&$W_{(34)}$&$\cdots$\\
			\hdashline
			$\vdots$&$\vdots$&$\vdots$&$\vdots$&$\ddots$\\
			\bottomrule
		\end{tabular}
	\end{table}
	
	\begin{itemize}
		\item On one hand, we find that within each row, there exists at most one $W_j$ from each group. As feature statistics $W_{j}$'s satisfy the between-group coin-flipping property, we have that $W_j$'s within the same row satisfy the coin-flipping property introduced in \citet{Candes2018}. Thus, we can estimate the false discovery proportion (FDP) of the rejection subset $\mathcal{R}^{(l)}(t)=\{j\in \mathcal{C}_l|W_j\geq t\}$ as
		\begin{equation}\label{G_FDR_naive}
			\widehat{\text{FDP}}^{(l)}(t)=\frac{1+\#\{j\in \mathcal{C}_l|W_j\leq -t\}}{1\vee\#\{j\in \mathcal{C}_l|W_j\geq t\}},
		\end{equation}
		for any $t>0$ and $l=1,2\ldots$ in the analogous way of \citet{Candes2018}. 
		\item On the other hand, we have $\text{FDP}^{(\text{fg})}(t)$ of $\mathcal{R}^{(\text{fg})}(t)$ is a weighted mean of $\text{FDP}^{(l)}(t)$'s of subsets $\mathcal{R}^{(l)}(t)=\mathcal{R}^{(\text{fg})}(t)\cap\mathcal{C}_l$ ($l=1,2,\ldots$) that 
		\begin{align*}
			\text{FDP}^{(\text{fg})}(t)
			=&\frac{\#(\mathcal{R}^{(\text{fg})}(t)\cap \mathcal{H}_0^{(\text{fg})})}{1\vee\#\mathcal{R}^{(\text{fg})}(t)}\\
			=&\sum_{l=1}^{\phi(t)}\frac{\#(\mathcal{R}^{(l)}(t)\cap \mathcal{H}_0^{(\text{fg})})}{1\vee\#\mathcal{R}^{(\text{fg})}(t)}\\=&\sum_{l=1}^{\phi(t)}\frac{1\vee\#\mathcal{R}^{(l)}(t)}{1\vee\#\mathcal{R}^{(\text{fg})}(t)}\times \frac{\#(\mathcal{R}^{(l)}(t)\cap \mathcal{H}_0^{(\text{fg})})}{1\vee\#\mathcal{R}^{(l)}(t)}\\
			=&\sum_{l=1}^{\phi(t)}\frac{1\vee\#\mathcal{R}^{(l)}(t)}{1\vee\#\mathcal{R}^{(\text{fg})}(t)}\times \text{FDP}^{(l)}(t),
		\end{align*}
		where the second equality comes from the fact that rejection subsets $\mathcal{R}^{(l)}(t)$ can be nonempty only for the first $\phi(t)=\max_k\#\{j\in B_k||W_j|\geq t\}$ rows because $|W_{(k1)}|\geq |W_{(k2)}|\geq \cdots$ in each group $B_k$.
	\end{itemize}

	As a result, using $\widehat{\text{FDP}}^{(l)}(t)$'s in (\ref{G_FDR_naive}), we can estimate $\text{FDP}^{(\text{fg})}(t)$ of $\mathcal{R}^{(\text{fg})}(t)$ by  
	\begin{align*}
		\widehat{\text{FDP}}^{(\text{fg})}(t)&=\sum_{l=1}^{\phi(t)}\frac{1\vee\#\mathcal{R}^{(l)}(t)}{1\vee\#\mathcal{R}^{(\text{fg})}(t)}\times \widehat{\text{FDP}}^{(l)}(t)\\&=\sum_{l=1}^{\phi(t)}\frac{1\vee\#\{j\in \mathcal{C}_l|W_j\geq t\}}{1\vee\#\{j|W_{j}\geq t \}}\times \frac{1+\#\{j\in \mathcal{C}_l|W_j\leq -t\}}{1\vee\#\{j\in \mathcal{C}_l|W_j\geq t\}}\\
		&=\frac{{\phi(t)}+\#\{j|W_{j}\leq -t \}}{1\vee\#\{j|W_{j}\geq t \}},
	\end{align*}
	leading to the naive filter (Algorithm \ref{alg:G-FDR}).
	
	\begin{algorithm}
		\caption{Naive feature filter with group knockoffs.}\label{alg:G-FDR}
		\begin{algorithmic}[1]
			\STATE  \textbf{Input:} Groups $B_1,\ldots,B_K$, feature statistics $W_1,\ldots,W_p$ and the target level $\alpha>0$.
			\STATE Compute
			{\begin{equation}\label{t_alpha1}
					t_\alpha=\min\Biggl\{t>0\Bigg|\widehat{\text{FDP}}^{(\text{fg})}(t)=\frac{{\phi(t)}+\#\{j|W_{j}\leq -t \}}{1\vee\#\{j|W_{j}\geq t \}}\leq \alpha\Biggr\},
			\end{equation}}
			\text{where }$\phi(t)=\max_k\#\{j\in B_k||W_j|\geq t\}$.
			\STATE \textbf{Output:} The rejection set $\mathcal{R}^{(\text{fg})}=\{j|W_j\geq t_\alpha\}$.
		\end{algorithmic}
	\end{algorithm}

	\subsection{FVG Filter}\label{filter2}
	
	Although the naive filter seems intuitive in FDR control, such a control can not be strictly proven. To fill the gap, we provide an alternative interpretation of Algorithm \ref{alg:G-FDR} and consider a variation that provides provable FDR control as follows.
	
	With the target to reject as many false $H_j^{(\text{fg})}$'s as possible while $\text{FDR}^{(\text{fg})}$ is under control, multiple testing of $H_j^{(\text{fg})}$'s can be transferred into an optimization problem. Based on the construction of feature statistics that $W_j$'s for those false $H_j^{(\text{fg})}$'s tend to be positive and large, Algorithm \ref{alg:G-FDR} solves the following optimization problem 
	$$\begin{aligned}
		\max |\mathcal{R}^{(\text{fg})}|,\quad
		\text{s.t. }& \widehat{\text{FDP}}^{(\text{fg})}=\sum_{l}\text{I}\left(\max_{j\in \mathcal{C}_l}|W_{j}|\geq t^{(l)}\right)\times\frac{1+\#\{j\in \mathcal{C}_l|W_{j}\leq -t^{(l)} \}}{1\vee[\sum_{l}\#\{j\in \mathcal{C}_l|W_{j}\geq t^{(l)} \}]}\leq \alpha,\\
		& \mathcal{R}^{(\text{fg})}=\cup_{l}\mathcal{R}^{(l)}=\cup_{l}\{j\in \mathcal{C}_l|W_j\geq t^{(l)}\},\quad t^{(1)}=t^{(2)}=\cdots.
	\end{aligned}$$
	Based on such an observation, we utilize theoretical results of \citet{Katsevich2019} that for any $t>0$, 
	$$\textbf{E}\left[\sup_{t>0}\frac{\#\{j\in \mathcal{H}_0^{(\text{fg})}\cap \mathcal{C}_l|W_{j}\geq t\}}{1+\#\{j\in\mathcal{H}_0^{(\text{fg})}\cap\mathcal{C}_l|W_{j}\leq -t\}}\right]\leq 1.93,$$ to modify the above optimization problem as 
	\begin{equation}\label{opt}
		\begin{aligned}[b]
			\max |\mathcal{R}^{(\text{fg})}|,\quad
			\text{s.t. }& \text{I}\left(\max_{j\in \mathcal{C}_l}|W_{j}|\geq t^{(l)}\right)\times\frac{1+\#\{j\in \mathcal{C}_l|W_{j}\leq -t^{(l)} \}}{1\vee[\sum_{l}\#\{j\in \mathcal{C}_l|W_{j}\geq t^{(l)} \}]}\leq \frac{v_l\alpha}{1.93},\quad l=1,2,\ldots,\\
			& \mathcal{R}^{(\text{fg})}=\cup_{l}\mathcal{R}^{(l)}=\cup_{l}\{j\in \mathcal{C}_l|W_j\geq t^{(l)}\},
		\end{aligned}
	\end{equation}
	where $v_l$'s are pre-calculated budgets for contribution of $\mathcal{R}^{(l)}$'s to $\widehat{\text{FDP}}^{(\text{fg})}$. By doing so, the rejection set {\rm $\mathcal{R}^{(\text{fg})}$} obtained by solving \eqref{opt} has controlled {\rm $\text{FDR}^{(\text{fg})}$} as shown in Theorem \ref{thm:FDR}.
	\begin{thm}\label{thm:FDR}
		The rejection set {\rm $\mathcal{R}^{(\text{fg})}$} obtained by solving \eqref{opt} controls {\rm $\text{FDR}^{(\text{fg})}$} at the target level  $\alpha>0$.
	\end{thm}
	The proof of Theorem \ref{thm:FDR} is provided in Appendix \ref{pr:FDR}. Observing that the optimization problem \eqref{opt} is analogous to the ones in (19) and (21) of \citet{Li2021} that
	\begin{itemize}
		\item[$\star$] feature statistics $W_j$'s within the each row of Table \ref{tab:align} satisfy the coin-flipping property while $W_j$'s for $j=1,\ldots,p$ do not (which is analogous to the $p\times p$ feature statistics matrix in Section 3.1 of \citet{Li2021}) and;
		\item[$\star$] the contribution of each row of Table \ref{tab:align} to $\widehat{\text{FDP}}^{(\text{fg})}$, measured by $$\text{I}\left(\max_{j\in \mathcal{C}_l}|W_{j}|\geq t^{(l)}\right)\times\frac{1+\#\{j\in \mathcal{C}_l|W_{j}\leq -t^{(l)} \}}{1\vee[\sum_{l}\#\{j\in \mathcal{C}_l|W_{j}\geq t^{(l)} \}]},$$
		is required to be constrained under the row-specific level $\frac{v_l\alpha}{1.93}$ with $\sum_l \frac{v_l\alpha}{1.93}=\frac{\alpha}{1.93}$,
	\end{itemize}
	we develop the feature-versus-group (FVG) filter (Algorithm \ref{alg:G-FDR2}), whose steps 4-10 adopt the threshold calculation strategy in Algorithm 2 of \citet{Li2021} to compute the optimal thresholds $t^{(l)}$'s. Similar to the discussion in supplementary note 4 of \citet{Sesia2020}, the correction factor $1.93$ is required in the proof for technical reasons, while empirical evidence suggests that it may be practically unnecessary. Thus, for all experiments in Sections \ref{Simulation}--\ref{Real}, we implement the FVG filter without the correction factor $1.93$. {For empirical reference, we compare the naive filter (Algorithm \ref{alg:G-FDR}), the FVG filter with and without the $1.93$ factor via experiments in Appendix \ref{Comparison2}.}
	
	\begin{algorithm}
		\caption{Feature-versus-group (FVG) filter with group knockoffs.}\label{alg:G-FDR2}
		\begin{algorithmic}[1]
			\STATE  \textbf{Input:} Groups $B_1,\ldots,B_K$, feature statistics $W_1,\ldots,W_p$ and the target level $\alpha>0$.
			\STATE Align $W_j$'s as shown in Table \ref{tab:align} such that $W_j$'s for all features in the group $B_k$ are in the $k$-th column and $|W_{(k1)}|\geq |W_{(k2)}|\geq \cdots$.
			\STATE Compute budget $v_l$ for the $l$-th row using $|W_1|,\ldots,|W_p|$ for $l=1,2,\ldots$ such that $\sum_{l=1}^{\infty}v_l=1$.
			\STATE Compute grids $\mathcal{G}_l=\{1/v_l,2/v_l,\ldots,({neg}_l+1)/v_l\}$ for the $l$-th row where
			$neg_l=\#\{j\in \mathcal{C}_l|W_j< 0\}$.
			\STATE Combine $\mathcal{G}_{comb}=(\cup_l \mathcal{G}_l)\cup \{0\}$ and sort values ${grid}_{(1)}> {grid}_{(2)}>\ldots$ in $\mathcal{G}_{comb}$. 
			\STATE Initialize $b=0$
			\REPEAT
			\STATE Update $b=b+1$.
			\STATE Compute 
			{\begin{equation}\label{t_alpha2}
					t^{(l)}=\min\Biggl\{t>0\Bigg|\frac{1+\#\{j\in \mathcal{C}_l|W_j\leq -t\}}{v_l}\leq {grid}_{(b)}\Biggr\},
			\end{equation}}
			for $l=1,2,\ldots$.
			\UNTIL{$\text{I}\left(\max_{j\in \mathcal{C}_l}|W_{j}|\geq t\right)\times\frac{1+\#\{j\in \mathcal{C}_l|W_{j}\leq -t^{(l)} \}}{1\vee[\sum_{l}\#\{j\in \mathcal{C}_l|W_{j}\geq t^{(l)} \}]}\leq \frac{v_l\alpha}{1.93}$ for all $l$.}
			\STATE \textbf{Output:} The rejection set $\mathcal{R}^{(\text{fg})}=\cup_{l}\{j\in  \mathcal{C}_l|W_{j}\geq t^{(l)}\}$.
		\end{algorithmic}
	\end{algorithm}
	
	We also provide other variations with details in Appendix \ref{extensions}.  
	
	\begin{rmk}
		In the FVG filter (Algorithm \ref{alg:G-FDR2}), the budgets $v_l$'s can be computed adaptively to the prior belief of the number of important features within each group. Here, we provide two strategies of specifying $v_l$'s.
		\begin{itemize}
			\item If it is believed that all important groups have similar effect size on the response $Y$, we suggest letting 
			$$v_l=\frac{\sum_{k}|W_{(kl)}|}{\sum_{l}\sum_{k}|W_{(kl)}|},\quad l=1,2,\ldots,$$
			such that contribution of $\mathcal{R}^{(l)}$'s to {\rm$\widehat{\text{FDP}}^{(\text{fg})}$} are constrained by the sum of signal magnitudes  $\sum_{k}|W_{(kl)}|$ in the $l$-th row of Table \ref{tab:align}.
			\item If it is believed that most of important groups have small effect size while a small number of important groups have large effect size, we suggest letting 
			$$v_l=\frac{\sum_{k}|W_{(kl)}|/l}{\sum_{l}\sum_{k}|W_{(kl)}|/l},\quad l=1,2,\ldots,$$
			such that more budgets are allocated to the first several rows of Table \ref{tab:align} and more important features in important groups with small effect sizes are identified.
		\end{itemize}
	\end{rmk}
	
	Unlike the FDR control, asymptotic power of our FVG filter can't be rigorously proved. The main reason is that our FVG filter is a model-free approach without rigid parametric assumption on the conditional distribution $Y|\textbf{X}$ and its power relies on the consistency between $Y|\textbf{X}$ and the feature importance score. However, when a  linear model is correctly used, our FVG filter can achieve asymptotic power $1$. For example, following the proof of \citet{Zou2006}, the feature importance scores $W_j$'s will be positive with asymptotic probability $1$ for those features with nonzero contribution to $Y$, making their asymptotic probability to be included in the selection set $\mathcal{R}^{(\text{fg})}$ as 1 under mild conditions.
	
	\section{Simulated Experiments with Real-world Genetic Data}\label{Simulation}
	
	To evaluate the proposed FVG filter in controlling FDR and identifying important features with comparison to existing approaches, we conduct extensive simulated experiments. 
	
	\subsection{Data Generation}\label{Sim_setting}
	
	We simulate datasets based on a real-world genetic data to mimic the dependency structure among features in real-world genetic analysis and investigate how the proposed method performs in real applications. Specifically, we simulate datasets based on $6095$ variants in the APOE/APOC region (chromosome 19, positions 44909011 $\sim$ 45912650) from the whole-genome sequencing data (NG00067.v5) of the Alzheimer’s Disease Sequencing Project \citep{Leung2019}. In our study, we restrict our focus to $6952$ individuals with estimated European ancestry rate\footnote{Estimated by SNPWeights v2.1 \citep{Chen2013} using reference populations from the 1000 Genomes Consortium \citep{10002015}.} greater than 80\% and $p=1157$ variants $X_1,\ldots,X_p$ whose minor allele frequency (MAF) is larger than $1\%$ and pairwise correlations are in $[-0.95,0.95]$ to avoid the existence of statistically indistinguishable variants with nearly-perfect correlations. By randomly sampling $n$ individuals without replacement, each simulated dataset is obtained by collecting the sampled individuals' variants and generating responses {\rm$y_1,\ldots,y_n$} from the linear model \begin{equation}\label{linear_model}
		Y=\sum_{j=1}^pX_{j}{\beta}_{j}+e,\quad e\sim \text{N}(0,4),
	\end{equation}
	where only $k=15$ randomly selected coefficients ${\beta}_{j}$'s are nonzero and follow $N(0,0.6^2)$. In our experiments, we simulate 1000 datasets for each possible sample size  $n=500,1000$ and $2000$. To circumvent the threshold phenomenon mentioned in \citet{Gimenez2019b}, we implement the multiple knockoff counterpart described in Appendix \ref{multiple}.
	
	\subsection{Details of Implementation and Evaluation Metrics}\label{Evaluation}
	
	Among these $1157$ variants, we compute the correlation {\rm$\text{cor}(X_i,X_{j})$} between each pair of variants $(X_i,X_{j})$ over all 6952 individuals and construct $345$ variant groups {\rm${B}^{(2)}_{1},\ldots,{B}^{(2)}_{345}$} by applying the hierarchical clustering (average linkage with cutoff value $0.5$) on the distance matrix {\rm$(1-|\text{cor}(X_i,X_{j})|)_{p\times p}$}. We define important features as variants $X_j$'s whose corresponding coefficients $\beta_{j}$'s are nonzero and important groups as those containing at least one important feature. Specifically, most important groups only contain one important feature. In summary, we have 
	\begin{itemize}
		\item $H_{j}^{(\text{ff})}$'s are false for 15 important features;
		\item $H_{k}^{(\text{gg})}$'s are false for groups containing important features;
		\item $H_{j}^{(\text{fg})}$'s are false for features in those groups with false $H_{k}^{(\text{gg})}$'s.
	\end{itemize}
	
	Based on groups {\rm${B}^{(2)}_{1},\ldots,{B}^{(2)}_{345}$} and the correlation matrix  {\rm$\boldsymbol{\Sigma}=(\text{cor}(X_i,X_{j}))_{1157\times 1157}$}, we generate second-order group knockoffs under the maximum entropy (ME) construction using Algorithm 2 of \citet{Chu2024} and compute feature statistics as $W_j=\text{sign}\left(T_{j}-\widetilde{T}_{j}\right)\cdot \max\left(T_{j},\widetilde{T}_{j}\right)$ with the joint lasso importance scores in the implementation of our FVG filter.

	To evaluate the power and precision of inference results produced by our FVG filter, we derive catching sets of signals based on the selection set $\mathcal{R}_{\text{FVG}}$ as 
	$$\mathcal{CC}_{\text{FVG}}=\{\mathcal{S}_{k,\text{FVG}}={B}^{(2)}_{k}\cap\mathcal{R}_{\text{FVG}}| k=1,\ldots,345\}.$$
	For comparison, we also implement several existing methods as follows.
	\begin{itemize}
		\item \textbf{Model-X knockoff filter:} Based on the correlation matrix  {\rm$\boldsymbol{\Sigma}$}, we generate second-order knockoffs at the variant level under the maximum entropy (ME) construction and use the lasso coefficient difference statistics \citep{Candes2018} in the implementation of model-X knockoff filter. Based on the selection set $\mathcal{R}_{\text{model-X}}$, we derive catching sets as 
		$$\mathcal{CC}_{\text{model-X}}=\{\mathcal{S}_{j,\text{model-X}}=\{j\}| j\in \mathcal{R}_{\text{model-X}}\}.$$
		\item \textbf{Group knockoff filter:} We implement the group knockoff filter based on the same {\rm${B}^{(2)}_{1},\ldots,{B}^{(2)}_{345}$} and the same second-order group knockoffs as our FVG filter. Based on the selection set $\mathcal{R}_{\text{group}}$, we derive catching sets as 
		$$\mathcal{CC}_{\text{group}}=\{\mathcal{S}_{k,\text{group}}={B}^{(2)}_{k}| k\in \mathcal{R}_{\text{group}}\}.$$
		\item \textbf{KnockoffZoom and KeLP:} For these methods which identify signals at multiple layers of different resolutions, we define three resolution layers as (a) layer 0: the level of variants; (b) layer 1: the level of variant groups ${B}^{(1)}_{1},\ldots,{B}^{(1)}_{650}$ obtained by the hierarchical clustering (average linkage with cutoff value $0.25$); (c) layer 2: the level of variant groups {\rm${B}^{(2)}_{1},\ldots,{B}^{(2)}_{345}$} which are also used in the FVG filter. Here, layer 2 corresponds to the same variant group level as the group knockoff filter while layer 1 is an intermediate layers between layer 2 and the level of variants. Based on such a hierarchical structure of variant groups that each variant group at high-resolution levels is nested under one group at low-resolution levels, we implement both KnockoffZoom and KeLP at layers $0$ and $2$ (2 layers) and layers $0$, $1$ and $2$ (3 layers). Thus, both KnockoffZoom and KeLP are tackling the same task as our FVG filter to refine those important variant groups in layer 2. For each layer, we generate second-order group knockoffs under the maximum entropy (ME) construction using Algorithm 2 of \citet{Chu2024}. After implementing either KnockoffZoom or KeLP, we obtain catching sets as the most specific discoveries among all discoveries in selection sets (same as Figure 1 of \citet{Gablenz2024}) as detailed in Appendix \ref{df_cc}. Specifically, to avoid obtaining empty catching sets, we implement KeLP in our article with carefully chosen tuning parameters as detailed in Appendix \ref{KeLP}.

	\end{itemize}

	In the following section, we evaluate the empirical FDR of our FVG filter via \eqref{FDR_H_ig}, empirical power of all approaches as the proportion of important features {(with nonzero coefficients)} included in at least one catching sets. We also measure the informativeness of different approaches by measuring the average size and purity\footnote{Purity refers to the minimum absolute correlation within the catching set.} of catching sets.

	\subsection{Performance Comparison of the FVG Filter with Existing Competitors}\label{Comparison}
	
	\begin{figure}[t]
		\centering
		\begin{minipage}{0.5\linewidth}
			\centering
			\includegraphics[width=\linewidth]{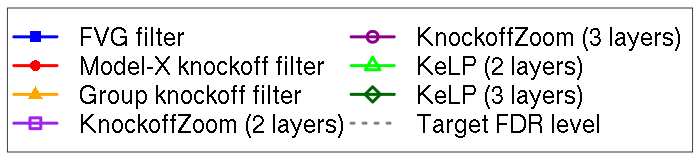}
		\end{minipage}\\
		\begin{minipage}{0.32\linewidth}
			\centering
			\includegraphics[width=\linewidth]{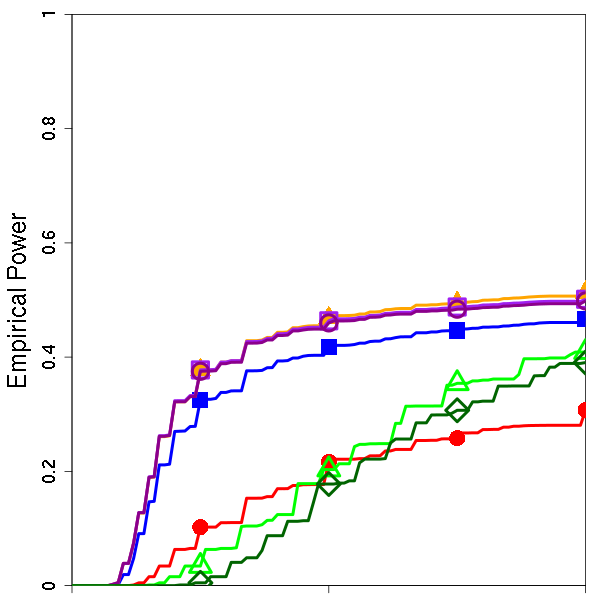}
			\includegraphics[width=\linewidth]{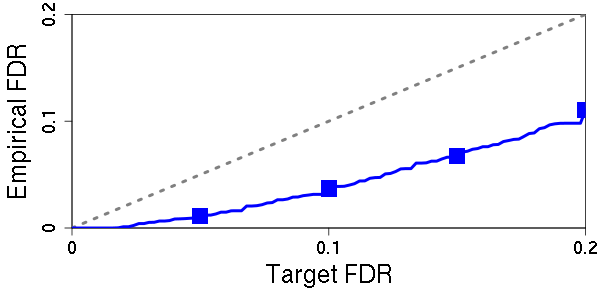}
			\par{(a) $n=500$.}
		\end{minipage}
		\begin{minipage}{0.32\linewidth}
			\centering
			\includegraphics[width=\linewidth]{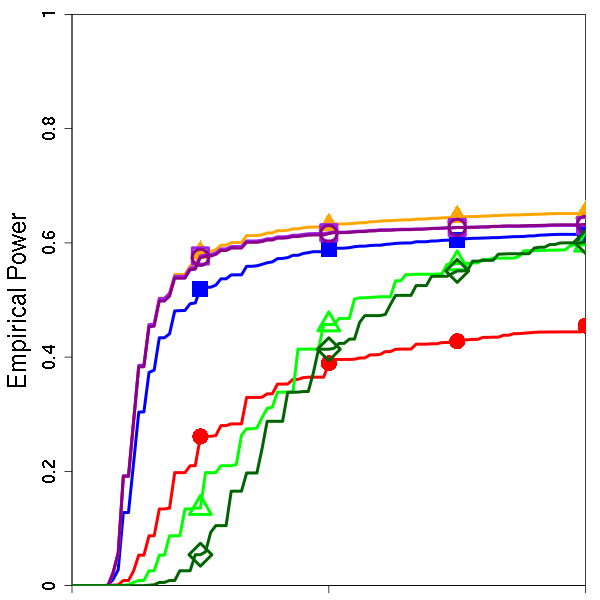}
			\includegraphics[width=\linewidth]{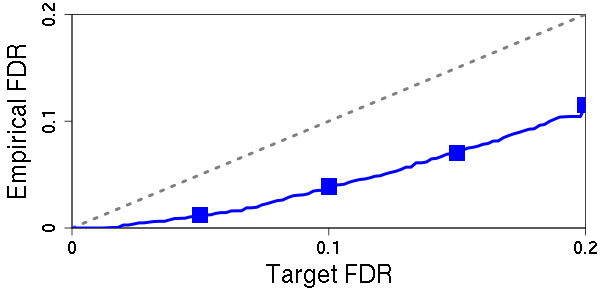}
			\par{(b) $n=1000$.}
		\end{minipage}
		\begin{minipage}{0.32\linewidth}
			\centering
			\includegraphics[width=\linewidth]{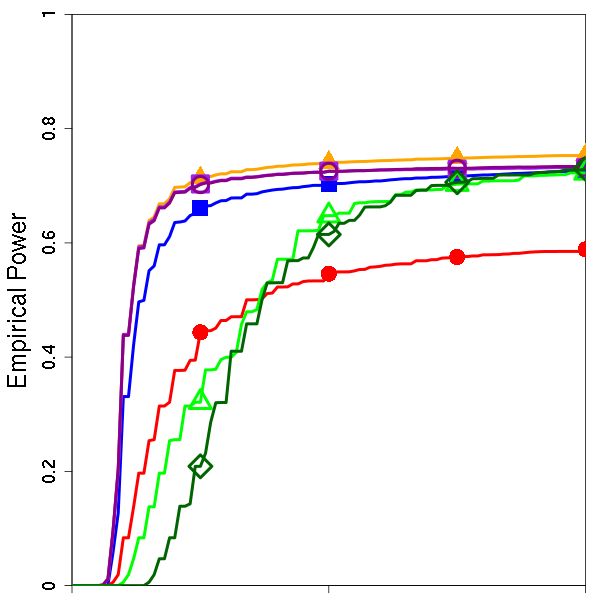}
			\includegraphics[width=\linewidth]{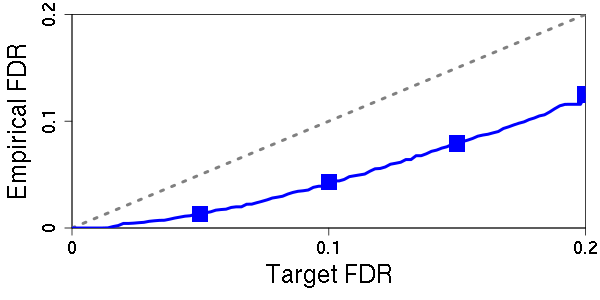}
			\par{(c) $n=2000$.}
		\end{minipage}
		\caption{Empirical FDR and power of FVG filter with respect to the target FDR level ($\alpha$) over 1000 simulated genetic datasets of different sample sizes in comparison to model-X knockoff filter \citep{Candes2018}, group knockoff filter \citep{Dai2016}, the KnockoffZoom \citep{Sesia2020} and the KeLP \citep{Gablenz2024}.}
		\label{fig:sample_size}
	\end{figure}

	\begin{figure}[t]
		\centering
		\begin{minipage}{0.32\linewidth}
			\centering
			\includegraphics[width=\linewidth]{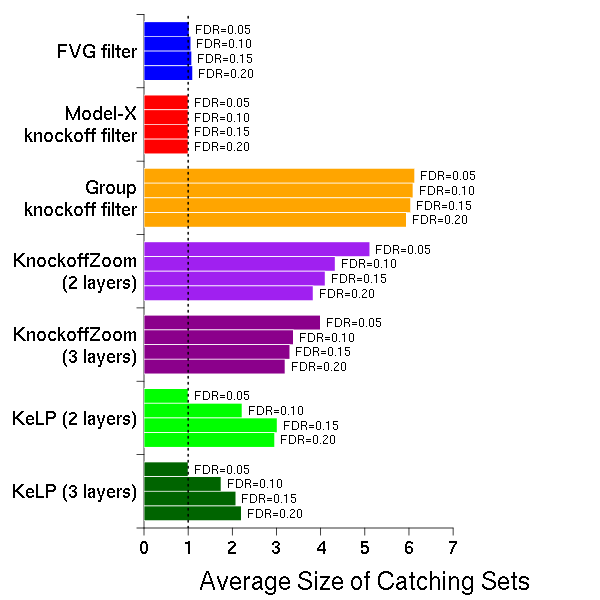}
			\includegraphics[width=\linewidth]{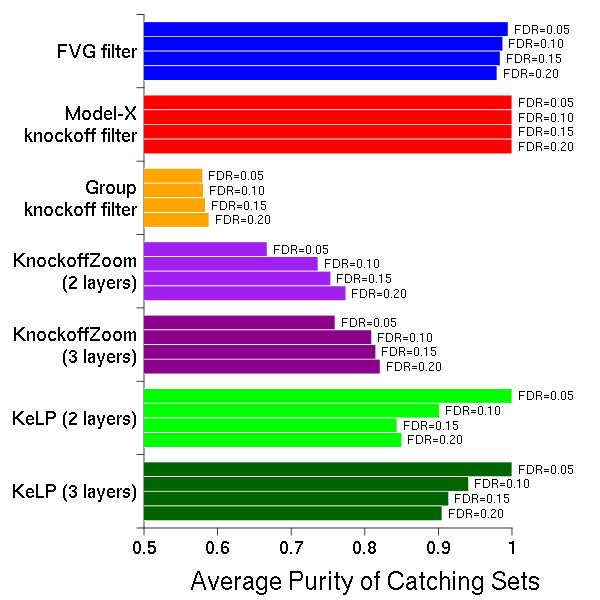}
			\par{(a) $n=500$.}
		\end{minipage}
		\begin{minipage}{0.32\linewidth}
			\centering
			\includegraphics[width=\linewidth]{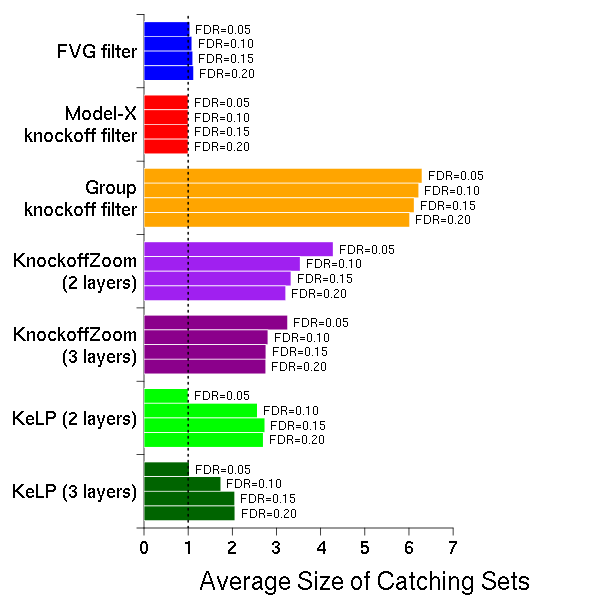}
			\includegraphics[width=\linewidth]{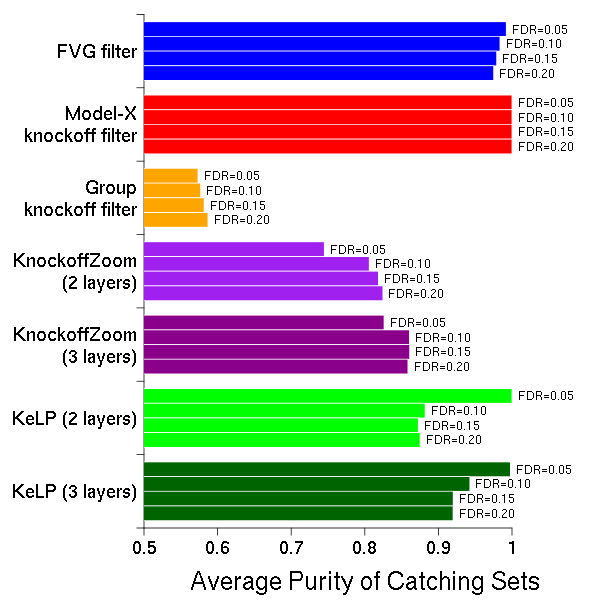}
			\par{(b) $n=1000$.}
		\end{minipage}
		\begin{minipage}{0.32\linewidth}
			\centering
			\includegraphics[width=\linewidth]{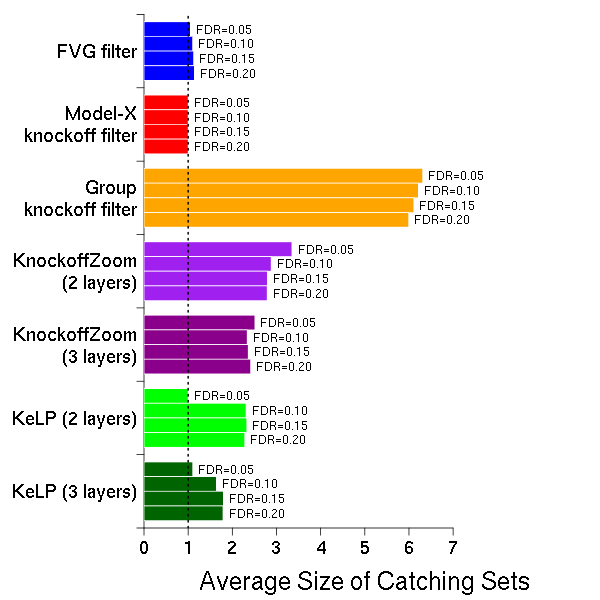}
			\includegraphics[width=\linewidth]{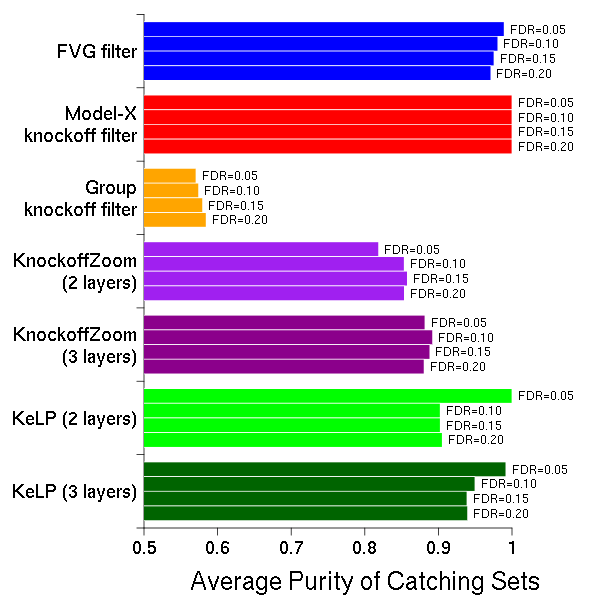}
			\par{(c) $n=2000$.}
		\end{minipage}
		\caption{Average sizes and purity of catching sets obtained FVG filter over 1000 simulated genetic datasets of different sample sizes in comparison to model-X knockoff filter \citep{Candes2018}, group knockoff filter \citep{Dai2016}, the KnockoffZoom \citep{Sesia2020} and the KeLP \citep{Gablenz2024}.}
		\label{fig:sample_size2}
	\end{figure} 
	
	Over 1000 simulated datasets of different sample sizes, the empirical FDR of our FVG filter and its power in comparison with all existing competitors under different target FDR levels are reported in Figure \ref{fig:sample_size}. We also present the average sizes and purity of catching sets produced by different methods in  Figure \ref{fig:sample_size2}. It is clear from Figure \ref{fig:sample_size} that the FVG filter controls $\text{FDR}^{(\text{fg})}$ at any target level $\alpha$. 
	
	Existing methods suffer different types of performance deterioration in identifying important features. 
	On one hand, the model-X knockoff filter suffers great power loss in Figure \ref{fig:sample_size} because high correlations of variants in the same group lower signal-to-noise ratio of false $H^{\text{(ff)}}_j$'s. On the other hand, although group knockoff filter  manages to achieve higher power in Figure \ref{fig:sample_size}, such an advantage comes from the sacrifice of informativeness of rejection sets. This can be seen from catching sets with average size of around $6$ and average purity of around $0.6$ in Figure \ref{fig:sample_size2} because variant groups are rejected as a whole without distinguishing which variants are more important. Such a poor precision of group knockoff filter can be somehow relieved by both KnockoffZoom and KeLP, which however still suffer lack of precision. Specifically, KnockoffZoom (2 layers) manages to maintain similar power as the group knockoff filter as shown in Figure \ref{fig:sample_size}. By decreasing the average size from $6$ to $3$ and increasing average purity from $0.6$ to $0.85$ when $n=2000$, the improvement of precision provided by KnockoffZoom is still limited. Adding the intermediate layer (layer 1) in KnockoffZoom only slightly improves precision of catching sets (average size of around $2.5$ and average purity of around $0.89$) without any power gain. KeLP, on the other hand, improves the precision of catching sets much greater than KnockoffZoom. With the same 2 layers in analysis, KeLP (2 layers) returns catching sets of average sizes of around $2.3$ and average purity of around $0.9$ when $n=2000$ (except for target FDR level $0.05$ where all catching sets are of size 1). Adding the intermediate layer (layer 1) could further refine catching sets  to  average size of $1.7$  and purity of  $0.95$. However, the power of KeLP is significantly lower especially under lower target FDR levels.
	
	{In contrast, by simultaneously selecting important groups and refining the catching sets, the FVG filter return catching sets that are smaller and purer without significant power loss in identifying important genetic variants. As sample size $n$ grows, FVG filter maintains catching sets of small size (around $1.1$) and high purity (higher than $0.97$) for target FDR levels $0.05$, $0.1$, $0.15$, $0.20$. With minor power loss with respect to group knockoff filter, our FVG filter can locate the real important variants from important groups in informative catching sets most of the time.  On the contrary, variant-level (layer 0) inferences of KnockoffZoom and KeLP still rely on the same variant-level knockoff of model-X knockoff filter. With low power at layer 0, a non-negligible fraction of catching sets are of size larger than $1$, leading to both large average size and lower average purity. Such an obstacle is elaborately circumvented by our FVG filter which uses group knockoffs to ease the identification of signals and lasso regression to refine the catching sets. }
	
	\subsection{{Multiple Signals in the Same Group}}\label{More_experiments}

	{In the simulation setting of Section \ref{Sim_setting}, there is only one signal within each nonnull group in most of time. As a result, heuristic identifying the contributing variants in nonnull groups are not that difficult. For example, one can always perform heuristic selection of contributing variants from variant groups identified by the group knockoff filter. However, such an identification task becomes nontrivial when there are multiple genetic signals in one variant group. To demonstrate the performance of the FVG filter in prioritizing important variants under such scenarios, we conduct more experiments with two settings of signal distributions, including:
		\begin{itemize}
			\item Setting A: $k=15$ selected coefficients ${\beta}_{j}$'s are nonzero and follow $N(0,0.6^2)$, where there is a randomly selected variant group containing 4 signals, two randomly selected  groups containing 2 signals and seven randomly selected  groups containing 1 signal;
			\item Setting B: $k=20$ selected coefficients ${\beta}_{j}$'s are nonzero and follow $N(0,0.6^2)$, where there are ten randomly selected  groups containing 2 signals.
		\end{itemize}
		Under both settings, we simulate 1000 datasets of sample size $n=2000$. We compare our FVG filter with two heuristic approaches.
		\begin{itemize}
			\item Top-1 signal: Within each variant group selected by the group knockoff filter, we construct the catching set by heuristically including only the variant with the largest $|W_j|$'s.
			\item SuSiE: Within each variant group selected by the group knockoff filter, we implement the SuSiE method \citep{Wang2020} and construct the catching set as the $95\%$ credible set.
	\end{itemize}}
	{Performance of our FVG filter and two heuristic approaches are presented in Figure \ref{fig:multiple}. We find that the FVG filter still achieves a balance between the informativeness of catching sets and power. On one hand, selecting only the top-1 signal within identified variant groups provides the most informative catching sets (with size 1 and purity 1). However, doing so may miss many signals, especially under setting B. On the other hand, using SuSiE within selected variant groups usually returns too large catching sets. In contrast, our FVG filter not only returns informative catching sets  (average size $< 1.5$ and average purity $\approx0.9$) but also has similar power as SuSiE. That is to say, among all genetic variants selected by SuSiE that are not top signals in the groups, our FVG filter adaptively selects those contributing variants without proxy ones.}
	
	{\begin{rmk}
			Although all the simulations are conducted under the linear model, the FDR control of our FVG filter still stands for other conditional distribution $Y|\textbf{X}$. The reason is that our FVG filter is developed under the similar framework of model-X knockoff, where no assumption is imposed on $Y|\textbf{X}$. Thus, even $Y|\textbf{X}$ follows a logistic regression model, using the lasso feature importance scores would not violate our FVG filter's FDR control. However, using a model that accurately depicts $Y|\textbf{X}$ to compute feature importance scores would further improve the power.
	\end{rmk}}

		\begin{figure}[t]
		\centering
		\begin{minipage}{0.9\linewidth}
			\centering
			\includegraphics[width=0.32\linewidth]{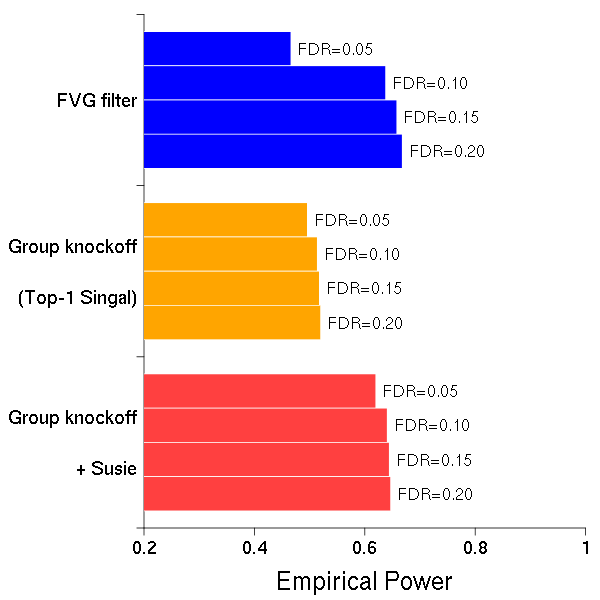}
			\includegraphics[width=0.32\linewidth]{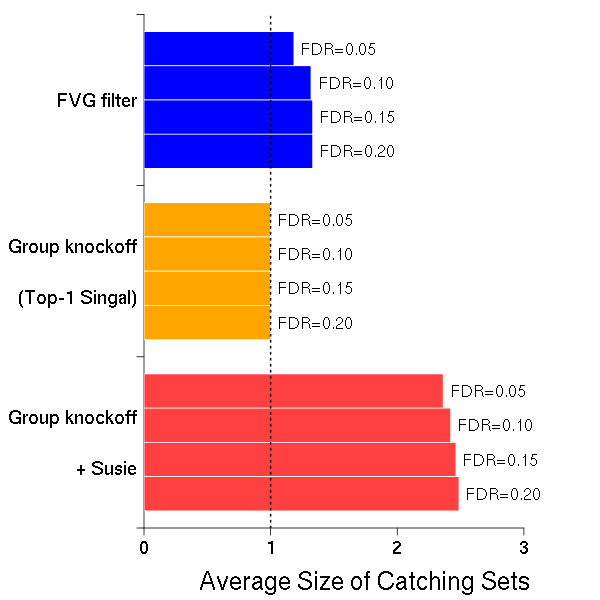}
			\includegraphics[width=0.32\linewidth]{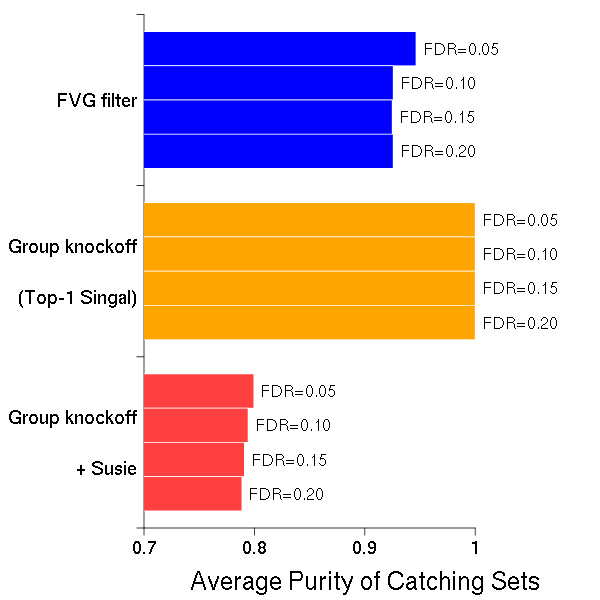}
			\par{(a) Setting A.}
		\end{minipage}
		\begin{minipage}{0.9\linewidth}
			\centering
			\includegraphics[width=0.32\linewidth]{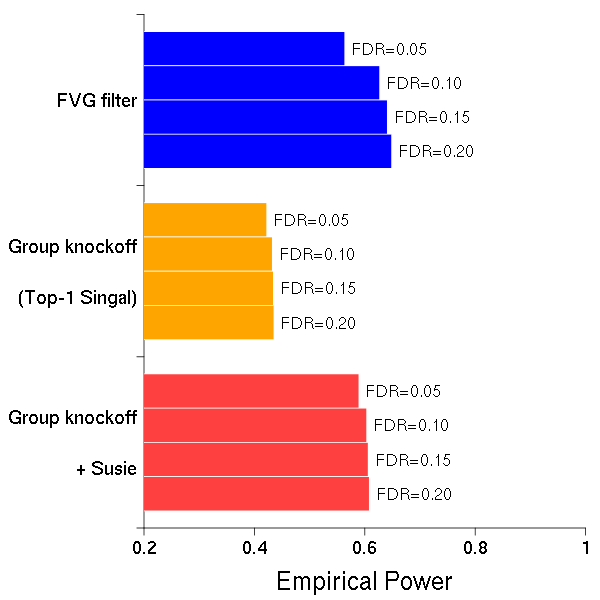}
			\includegraphics[width=0.32\linewidth]{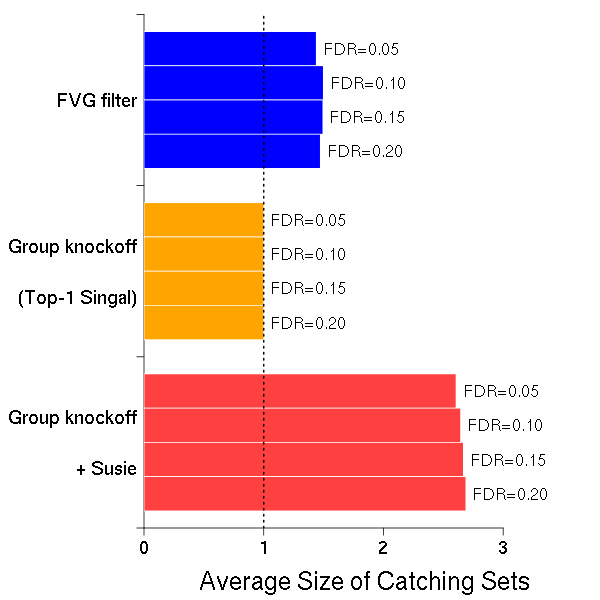}
			\includegraphics[width=0.32\linewidth]{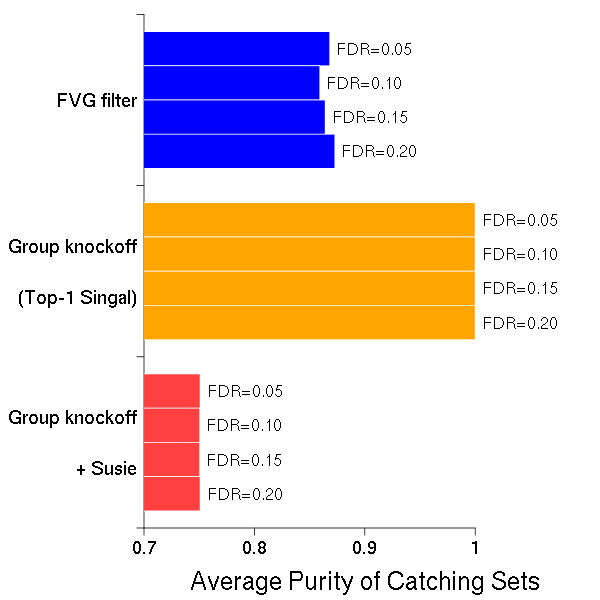}
			\par{(b) Setting B.}
		\end{minipage}
		\caption{Average power, sizes and purity of catching sets obtained by FVG filter over 1000 simulated genetic datasets of different settings in comparison to two heuristic approaches.}
		\label{fig:multiple}
	\end{figure} 
	
	\section{AD Analysis of EADB--UKBB Dataset}\label{Real}
	
	Given the promising performance of the proposed filter in simulated experiments, we apply it to the EADB-UKBB dataset where model-X knockoff filter and KeLP exhibits power loss and the group knockoff filter and KnockoffZoom returns not informative enough results in Section \ref{Existing}.

	To infer which variants are associated with AD, we first compute correlation $\text{cor}(G_i,G_j)$ between any pair of variants $G_i$ and $G_j$ using the UK Biobank directly genotyped data as the reference panel and construct variant groups $B_1,\ldots,B_K$ by applying the hierarchical clustering (average linkage with cutoff value $0.5$) on the distance matrix {\rm$(1-|\text{cor}(G_i,G_{j})|)_{p\times p}$}. By doing so, we obtain the same $321,569$ groups $B_1,\ldots,B_k$ as the group knockoff filter in Section \ref{Existing} where the average absolute correlation of variants from different groups are lower than 0.5. Based on the correlation matrix {\rm$(\text{cor}(G_i,G_{j}))_{p\times p}$} and variant groups $B_1,\ldots,B_K$, we generate second-order group knockoffs $\widetilde{\mathbb{G}}$ under the maximum entropy (ME) construction using Algorithm 2 of \citet{Chu2024} by assuming that the distribution of $\textbf{X}$ can be well approximated by a multivariate Gaussian distribution. Given the response of interest $\textbf{y}=(y_1,\ldots,y_n)^\dT$ where $y_i=1$ if the $i$-th observation corresponds to a clinically diagnosed AD case and $y_i=0$ otherwise, we use absolute values of lasso estimators of the logistic regression model,
	$$
	\text{logit}\{\text{Pr}(Y=1|{\textbf{G}},\widetilde{\textbf{G}})\}=\beta_0+\sum_{j=1}^p\left\{{G}_{j}{\beta}_{j}+\widetilde{G}_{j}\widetilde{\beta}_{j}\right\}$$ as feature importance scores. In other words, we compute  $T_j=\left|\widehat{\beta}_j\right|$ and $\widetilde{T}_j=\left|\widehat{\widetilde{\beta}}_j\right|$, which are also used in other knockoff-based methods and provide more informative rejection sets without large power loss given the large sample size. We then calculate feature statistics $W_j=\text{sign}\left(T_{j}-\widetilde{T}_{j}\right)\cdot \max\left(T_{j},\widetilde{T}_{j}\right)$
	and implement the FVG filter  with target FDR level $\alpha=0.10$. 
	
	\begin{figure}
		\centering
		\includegraphics[width=0.9\linewidth]{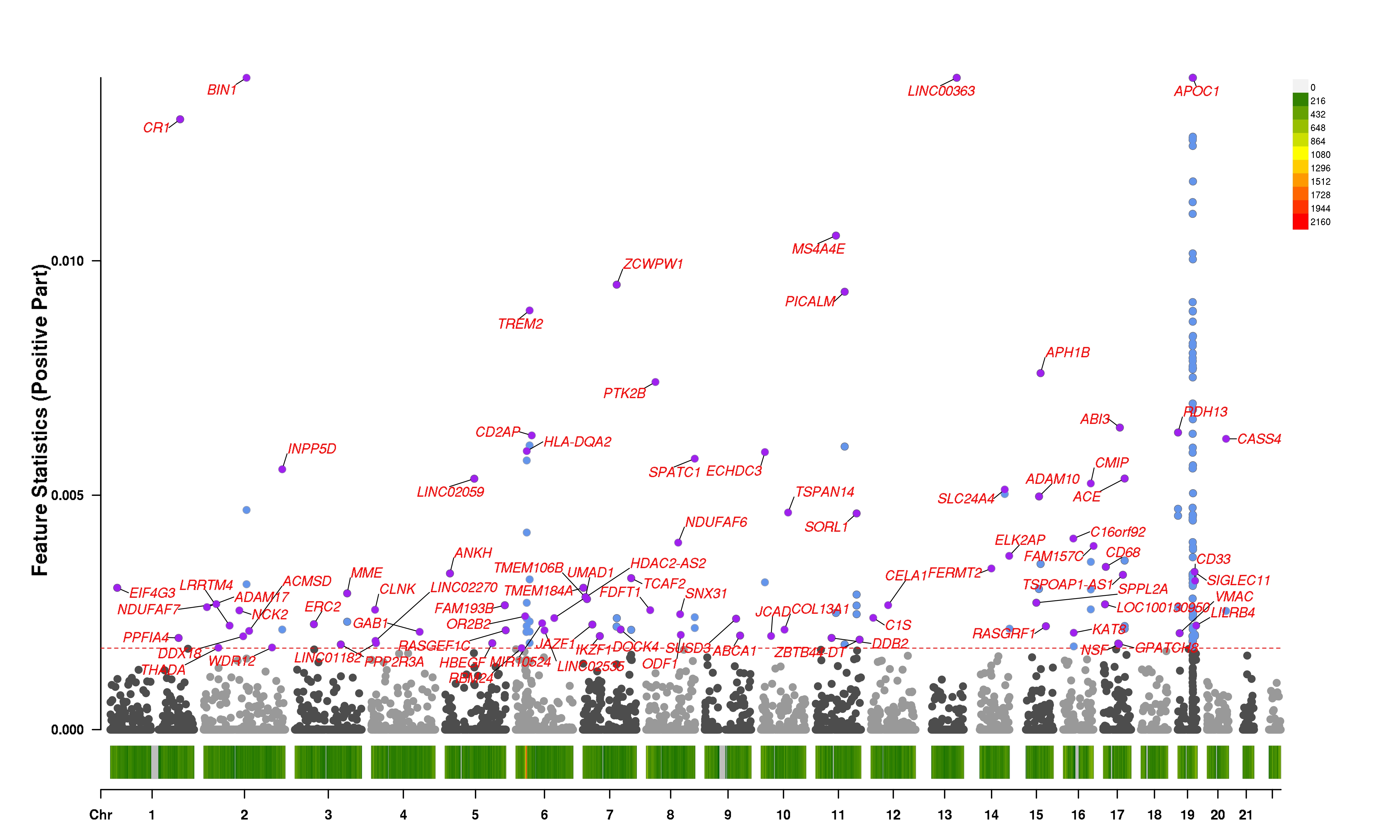}
		\caption{Manhattan plot of genetic variants under our FVG filter. Identified variants under target FDR level $\alpha=0.10$ are highlighted with names of their closest genes.}
		\label{fig:Real_Manhatton}
	\end{figure}
	
	\renewcommand{\arraystretch}{1.1}
	\begin{table}
		\centering
		\caption{Details of variants identified by the FVG filter  under the EADB--UKBB dataset and target FDR level $\alpha=0.10$. {Here, the column ``cS2Gene" denotes the gene that each genetic variant functionally annotates, and the column ``Annotation" depicts the type of annotation.}}\label{Real_table1}
		\resizebox{\columnwidth}{!}{
			\begin{tabular}{lrrrrlrrrrlrrr}
				\toprule
				Identified&\multirow{2}{*}{Position}&\multirow{2}{*}{cS2Gene}&\multirow{2}{*}{Annotation}&&Identified&\multirow{2}{*}{Position}&\multirow{2}{*}{cS2Gene}&\multirow{2}{*}{Annotation}&&Identified&\multirow{2}{*}{Position}&\multirow{2}{*}{cS2Gene}&\multirow{2}{*}{Annotation}\\
				Variant&&&&&Variant&&&&&Variant&&&\\\midrule
				\multicolumn{4}{c}{Chromosome 1}&&rs3740890&130385218&&&&rs419010&44865063&PVRL2&GTeX\\\cline{1-4}\cline{6-9}
				rs2305463& 20853688&EIF4G3&Exon&&\multicolumn{4}{c}{Chromosome 12}&&rs365653&44858389&&\\\cline{6-9}
				rs782791&203029691&PPFIA4&GTeX,ABC&&rs7183& 7070715&C1S&Exon&&rs17561351&44869072&TOMM40&eQTLGen,ABC\\
				rs4844610&207629207&CR1&GTeX&&rs17860282&51346566&CELA1&Promoter&&rs73052307&44881148&APOE&EpiMap\\\cline{1-4}\cline{6-9}
				\multicolumn{4}{c}{Chromosome 2}&&\multicolumn{4}{c}{Chromosome 13}&&rs11668327&44895376&&\\\cline{1-4}\cline{6-9}
				rs35933838&  9544047&CPSF3&eQTLGen&&rs9516168&93085490&&&&rs449647&44905307&APOE&Promoter,ABC\\\cline{6-9}
				rs6726709& 37270395&&&&\multicolumn{4}{c}{Chromosome 14}&&rs41290122&44880326&&\\\cline{6-9}
				rs77059113& 43445369&&&&rs17125924& 52924962&&&&rs10420036&44882772&APOE&EpiMap\\
				rs113068192& 76786226&&&&rs12590654& 92472511&&&&rs283808&44883777&PVRL2&GTeX\\
				rs116038905&105805908&&&&rs4904929& 92470949&&&&rs283813&44885917&PVRL2&Exon\\
				rs74851408&117201103&&&&rs74093831&105733666&CRIP1&EpiMap&&rs7254892&44886339&APOE&EpiMap=0.500\\
				rs6733839&127135234&ERCC3&Cicero&&rs873533&106667442&&&&rs7412&44908822&APOE&Exon\\\cline{6-9}
				rs744373&127137039&BIN1&GTeX&&\multicolumn{4}{c}{Chromosome 15}&&rs445925&44912383&APOE&GTeX\\\cline{6-9}
				rs2118506&127139927&BIN1&EpiMap&&rs12595082&50715532&&&&rs390082&44913574&APOC1&Promoter,EpiMap\\
				rs4954187&134842848&&&&rs2043085&58388755&LIPC&GTeX&&rs190712692&44921921&&\\
				rs6725887&202881162&&&&rs593742&58753575&&&&rs141622900&44923535&&\\
				rs10933431&233117202&NGEF&GTeX&&rs16946801&63312881&CA12&GTeX&&rs41290100&44867684&&\\
				rs7421448&233117495&INPP5D&EpiMap&&rs117618017&63277703&APH1B&Exon,GTeX,eQTLGen,ABC&&rs187183066&44872328&CEACAM22P&GTeX\\\cline{1-4}
				\multicolumn{4}{c}{Chromosome 3}&&rs56117790&79075361&&&&rs41290108&44874585&PVRL2&Promoter\\\cline{1-4}\cline{6-9}
				rs4974180& 56200107&&&&\multicolumn{4}{c}{Chromosome 16}&&rs73033507&44928146&APOC4&EpiMap\\\cline{6-9}
				rs7634492&136105288&SLC35G2&GTeX&&rs12325539&30022312&DOC2A&Promoter&&rs79149284&45010596&DMWD&GTeX\\
				rs16824536&155069722&&&&rs11865499&31120929&KAT8&GTeX,eQTLGen&&rs77196615&44877078&ZNF296&EpiMap\\
				rs61762319&155084189&MME&Exon&&rs78924645&31143037&PRSS36&Promoter,GTeX&&rs79701229&44881674&PVRL2&Promoter\\\cline{1-4}
				\multicolumn{4}{c}{Chromosome 4}&&rs8058370&81738205&PLCG2&Promoter&&rs3745150&44882502&TOMM40&GTeX\\\cline{1-4}
				rs6448453& 11024404&&&&rs12446759&81739398&PLCG2&Promoter,Cicero&&rs157580&44892009&TOMM40&Promoter\\
				rs13116567& 12239152&&&&rs11548656&81883307&PLCG2&Exon&&rs2075649&44892073&TOMM40&Promoter,GTeX,eQTLGen\\
				rs502746& 14019593&&&&rs1071644&81937798&PLCG2&Exon&&rs8106922&44898409&&\\
				rs116791081&143428212&&&&rs56407236&90103687&&&&rs405697&44901434&TOMM40&Exon\\\cline{1-4}\cline{6-9}
				\multicolumn{4}{c}{Chromosome 5}&&\multicolumn{4}{c}{Chromosome 17}&&rs7259620&44904531&APOC1&GTeX\\\cline{1-4}\cline{6-9}
				rs25992& 14707491&ANKH&Exon,Cicero&&rs7225151& 5233752&SCIMP&GTeX,ABC,Cicero&&rs405509&44905579&APOE&Promoter,ABC\\
				rs10068419& 86923485&&&&rs9901675& 7581494&CD68&Exon&&rs440446&44905910&APOE&Exon,ABC\\
				rs7268&140332965&HBEGF&Exon&&rs79759095&44451394&&&&rs142042446&44883210&APOE&EpiMap\\
				rs335473&177559423&RAB24&EpiMap&&rs199451&46724418&KANSL1-AS1&GTeX&&rs12972970&44884339&TOMM40&ABC\\
				rs4700716&180138792&RASGEF1C&Promoter&&rs616338&49219935&ABI3&Exon,EpiMap&&rs283815&44887076&&\\\cline{1-4}
				\multicolumn{4}{c}{Chromosome 6}&&rs2632516&58331728&BZRAP1&ABC&&rs6857&44888997&PVRL2&Exon\\\cline{1-4}
				rs78799498& 17319361&RBM24&ABC&&rs9896864&63458947&&&&rs71352238&44891079&TOMM40&Promoter,GTeX,ABC\\
				rs2130357& 27919052&PGBD1&GTeX&&rs4277405&63471557&ACE&GTeX&&rs184017&44891712&TOMM40&Promoter,ABC\\
				rs3132451& 31614248&AIF1&Promoter,EpiMap&&rs4309&63482562&ACE&Exon,EpiMap&&rs2075650&44892362&&\\
				rs532965& 32610196&HLA-DRB5&ABC&&rs3730025&63480412&ACE&Exon&&rs157581&44892457&TOMM40&Exon\\\cline{6-9}
				rs9469112& 32447376&&&&\multicolumn{4}{c}{Chromosome 19}&&rs34095326&44892587&BCAM&GTeX\\\cline{6-9}
				rs9271162& 32609938&HLA-DRB1&EpiMap&&rs4147918& 1058177&ABCA7&Exon&&rs34404554&44892652&&\\
				rs9273349& 32658092&&&&19:1051137\_CTG\_C& 1051138&ABCA7&Exon,GTeX,EpiMap,ABC&&rs157582&44892962&TOMM40&Promoter\\
				rs9275184& 32686937&HLA-DQA2&EpiMap&&rs3752246& 1056493&ABCA7&Exon,GTeX&&rs10119&44903416&TOMM40&Exon,GTeX\\
				rs3997854& 32715138&TAP2&GTeX&&rs35917007& 1849148&REXO1&Promoter,GTeX&&rs769449&44906745&APOE&Promoter,ABC\\
				rs114812713& 41066261&OARD1&Exon&&rs10404195& 5908958&VMAC&Exon&&rs75627662&44910319&SNRPD2&GTeX\\
				rs7748513& 41160234&TREML2&GTeX&&rs60184524&43790063&&&&rs7256200&44912678&APOC1&GTeX,EpiMap,ABC\\
				rs2093395& 41187288&TREML2&EpiMap&&rs445752&43956053&&&&rs438811&44913484&APOC1&Promoter,ABC\\
				rs143332484& 41161469&TREM2&Exon,EpiMap&&rs2125579&44288548&ZNF235&Exon,GTeX,eQTLGen&&rs12721046&44917997&&\\
				rs9394778& 41247320&TREML2&EpiMap&&rs8107530&44332247&&&&rs12721051&44918903&APOC1&Exon\\
				rs3997700& 41251889&&&&rs62116778&44389907&&&&rs111789331&44923868&&\\
				rs2171089& 47548075&CD2AP&GTeX&&rs151111866&44387243&ZNF285&Exon&&rs73936968&44892559&&\\
				rs117591349& 78478451&&&&rs2571174&44430314&ZNF229&Exon&&rs439401&44911194&APOE&GTeX\\
				rs313161& 85332752&&&&rs17800789&44510962&&&&rs59325138&44913034&APOC1P1&GTeX,Cicero\\
				rs976271&114361563&&&&rs7254925&44513811&&&&rs584007&44913221&APOE&GTeX\\\cline{1-4}
				\multicolumn{4}{c}{Chromosome 7}&&rs74602839&44481157&&&&rs157595&44922203&&\\\cline{1-4}
				rs10952168&  1554792&TMEM184A&Promoter,GTeX,ABC&&rs62116891&44524436&&&&rs1064725&44919304&APOC1&Exon\\
				rs12669393&  7815748&&&&rs930461&44539272&CEACAM22P&Exon&&rs118060185&44928417&ZNF155&GTeX\\
				rs3173615& 12229791&TMEM106B&Exon&&rs16979269&44552920&ZNF180&GTeX&&rs144311893&44920687&&\\
				rs2189965& 28132395&&&&rs56757528&44561286&&&&rs4803771&44924391&EML2&GTeX\\
				rs4917014& 50266267&IKZF1&EpiMap,Cicero&&rs10405030&44565601&&&&rs60049679&44926451&APOC1P1&Promoter\\
				rs1476679&100406823&LAMTOR4&GTeX&&rs2357100&44595918&&&&rs114533385&44933496&ZNF226&GTeX\\
				rs41280968&100105324&AP4M1&Exon,ABC&&rs7259313&44600383&&&&rs79429216&44942260&APOC4&Exon,GTeX,EpiMap\\
				rs221834&100745552&UFSP1&Cicero&&rs17658507&44617986&PVR&GTeX&&rs5167&44945208&APOC4&Exon\\
				rs144867634&111940111&DOCK4&Exon&&rs80122548&44618415&&&&rs3760627&44953923&CLPTM1&Promoter\\
				rs12703526&143410495&EPHA1-AS1&ABC&&rs62119261&44618732&APOC4&EpiMap&&rs16979595&44974124&&\\
				rs7805776&143427203&TAS2R60&eQTLGen&&rs203717&44636512&IGSF23&Exon&&rs73560271&45051621&ZNF296&EpiMap\\\cline{1-4}
				\multicolumn{4}{c}{Chromosome 8}&&rs2965160&44692865&CEACAM19&eQTLGen&&rs78620885&45087826&PPP1R37&ABC\\\cline{1-4}
				rs4731& 11808828&FDFT1&Exon,ABC&&rs2965164&44698782&CEACAM16&Promoter&&rs117106519&45047229&CLASRP&ABC\\
				rs73223431& 27362470&PTK2B&GTeX,eQTLGen&&rs62120574&44711506&&&&rs3178166&45090912&GEMIN7&Exon,GTeX,eQTLGen\\
				rs9297949& 94957217&NDUFAF6&Promoter&&rs76419583&44714590&CEACAM22P&GTeX&&rs10405086&45123977&&\\
				rs1693551&100663356&SNX31&Exon,GTeX&&rs2965109&44722081&BCL3&ABC&&rs2532037&45194545&&\\
				rs2436889&102564430&&&&rs1004165&44728939&&&&rs3848526&45195982&SYMPK&GTeX\\
				rs79832570&144042819&PARP10&GTeX&&rs2965169&44747899&BCL3&Exon,ABC&&rs10411314&45224801&EXOC3L2&Exon\\
				rs61732533&144053248&OPLAH&Exon&&rs116480598&44753996&BCL3&ABC&&rs17356664&45237513&&\\
				rs34674752&144099319&SHARPIN&Exon,ABC=0.500&&rs114294368&44789899&&&&rs77740243&45248200&&\\\cline{1-4}
				\multicolumn{4}{c}{Chromosome 9}&&rs75599500&44740439&BCL3&EpiMap&&rs11559024&45317925&CKM&Exon\\\cline{1-4}
				rs1131773& 93077974&SUSD3&Exon,GTeX,EpiMap&&rs2965174&44741758&CBLC&GTeX&&rs17875640&45323890&CKM&Promoter,ABC\\
				rs1883025&104902020&&&&rs10408993&44802496&&&&rs73047784&45458317&&\\\cline{1-4}
				\multicolumn{4}{c}{Chromosome 10}&&rs117326714&44783846&&&&rs3916883&45352060&ERCC2&ABC\\\cline{1-4}
				rs7920721&11678309&USP6NL&EpiMap,ABC,Cicero&&rs28399654&44813331&BCAM&Exon&&rs4802253&45401737&PPP1R13L&ABC\\
				rs7068614&11679478&&&&rs117142879&44802327&&&&rs2336219&45409148&PPP1R13L&GTeX\\
				rs117980033&29966853&&&&rs28399637&44820881&BCAM&Promoter,GTeX,EpiMap,ABC&&rs7256865&45515149&RTN2&eQTLGen\\
				rs1227759&69831131&COL13A1&GTeX,eQTLGen&&rs10405693&44823407&BCAM&EpiMap&&rs73942917&45543797&DMWD&Cicero\\
				rs6586028&80494228&TSPAN14&GTeX,ABC&&rs149529419&44868071&&&&rs1264226&45559909&OPA3&GTeX\\\cline{1-4}
				\multicolumn{4}{c}{Chromosome 11}&&rs140684051&44896199&&&&rs73568973&45549792&OPA3&Exon\\\cline{1-4}
				rs1685404& 47222114&&&&rs73572039&44839373&PVRL2&GTeX&&rs73570937&45588621&GPR4&EpiMap\\
				rs1582763& 60254475&&&&rs76692773&44890954&TOMM40&Promoter,ABC&&rs75987583&45745299&DMWD&GTeX\\
				rs11824773& 60309467&MS4A6E&Cicero&&rs61642202&44842026&ZNF224&GTeX&&rs875121&49948251&&\\
				rs536841& 86076782&&&&rs117310449&44890259&TOMM40&Promoter,ABC&&rs2075803&51125272&SIGLEC9&Exon,eQTLGen\\
				rs10792832& 86156833&PICALM&ABC&&rs77301115&44893716&TOMM40&Promoter,eQTLGen&&rs1710398&51223655&CD33&GTeX\\
				rs3844143& 86139201&EED&EpiMap&&rs79398853&44895528&&&&rs3865444&51224706&CD33&Promoter\\
				rs74685827&121482368&SORL1&ABC,Cicero&&rs115881343&44899959&TOMM40&eQTLGen&&rs731170&54664811&LILRB4&Exon\\\cline{11-14}
				rs2298813&121522975&SORL1&Exon&&rs41289512&44848259&CLASRP&Cicero&&\multicolumn{4}{c}{Chromosome 20}\\\cline{11-14}
				rs3781837&121578263&&&&rs12691088&44915229&APOC1&Exon,ABC&&rs6064392&56409712&CASS4&GTeX,eQTLGen,EpiMap,ABC\\
				rs12272618&121589615&SORL1&Promoter&&rs3865427&44877704&PVRL2&GTeX&&rs6069749&56446519&CASS4&GTeX\\
				\bottomrule
		\end{tabular}}
	\end{table}

	Variants in the rejection set $\mathcal{R}^{(\text{fg})}$ under target FDR level $\alpha=0.10$ are presented in Figure \ref{fig:Real_Manhatton} with their positive part of importance scores $W_j\cdot \text{I}(W_j>0)$'s and names of the closest genes.  Details of these variants are provided in Table \ref{Real_table1}.  In total, we identify 266 variants that contribute to AD variation substantially, which comes from 198 groups in 89 loci.  Similar to the literature, multiple variants are identified in the APOE/APOC region with the strongest association to AD (chromosome 19, positions 44909011 $\sim$ 45912650). In addition, we also manage to identify variants \texttt{rs6733839}, \texttt{rs744373} and \texttt{rs2118506} close to gene ``BIN1", variant \texttt{rs4844610} close to gene ``CR1", variants \texttt{rs536841}, \texttt{rs10792832} and \texttt{rs3844143} close to gene ``PICALM", variants \texttt{rs4277405}, \texttt{rs4309} and \texttt{rs3730025} close to gene ``ACE", all of which are also reported in  \citet{He2022} and \citet{Bellenguez2022}. Compared with the model-X knockoff filter, the FVG filter  only misses 2 loci identified by the model-X knockoff filter but discover 53 more loci. The advantage of our FVG filter over KeLP is more dominating.  No matter whether the intermediate layer is included in inference of KeLP, all KeLP-identified loci are discovered by our FVG filter while our FVG filter also returns at least another 44 new loci. This suggests that our method would not suffer from the same power loss issue of the model-X knockoff filter and KeLP. Compared to the group knockoff filter and KnockoffZoom\footnote{KnockoffZoom's result at the group-versus-group level is the same as the group knockoff filter, no matter whether the intermediate layer is included in inference.} which identify 198 AD-associated variant groups from 91 loci under the same target FDR level $0.10$, our FVG filter does not suffer power loss. As shown in Figure \ref{fig:Real} (a), among all variant groups identified by either the group knockoff filter (and KnockoffZoom) or our FVG filter\footnote{Here, we say a variant group is identified by our FVG filter if it contains at least one variant identified by our FVG filter.}, 	192 groups are identified by both methods while both our FVG filter and the group knockoff filter identify 6 groups not in common. This suggests the consistency of inference results between our FVG filter  and the group knockoff filter (and KnockoffZoom). However, our FVG filter manages to detect signals at variant \texttt{rs7225151} close to gene ``LOC100130950", while none of variant groups in the ``LOC100130950" locus are identified by the group knockoff filter (and KnockoffZoom). 

	{

	}
	\begin{figure}
		\centering
		\begin{minipage}[t]{0.48\linewidth}
			\centering
			\includegraphics[width=1\linewidth]{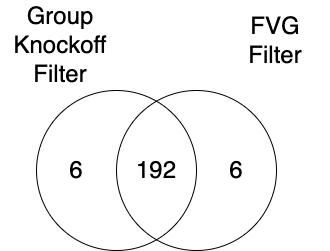}
			\par{(a) Venn diagram of  variant groups identified by group knockoff filter and our FVG filter.}
		\end{minipage}
		\begin{minipage}[t]{0.48\linewidth}
			\centering
			\includegraphics[width=1\linewidth]{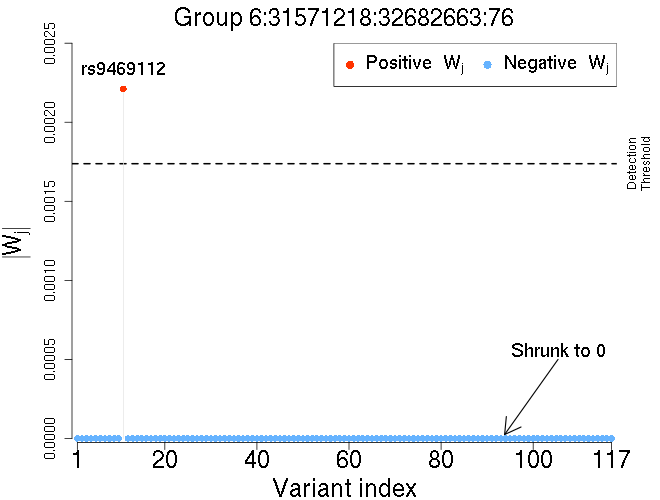}
			\par{{(b) Manhattan plot of feature importance statistics of different genetic variant  in a high-LD region between positions 31571218 $\sim$ 32682663,  chromosome 6, under the FVG filter.}}
		\end{minipage}\\ 
		\begin{minipage}{0.48\linewidth}
			\centering
			\includegraphics[width=1\linewidth]{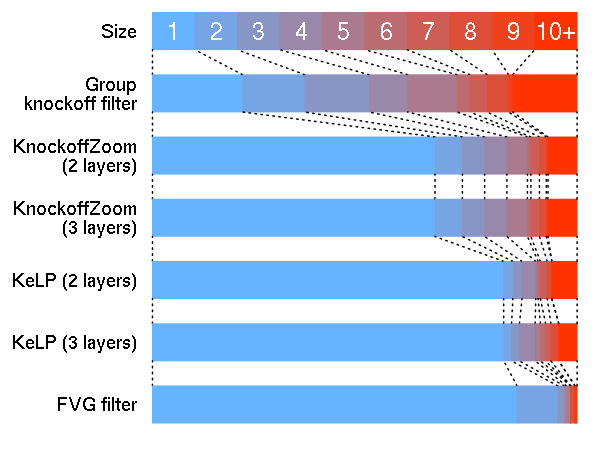}
			\par{(c) Comparison of catching set size among different approaches.}
		\end{minipage}
		\begin{minipage}{0.48\linewidth}
			\centering
			\includegraphics[width=1\linewidth]{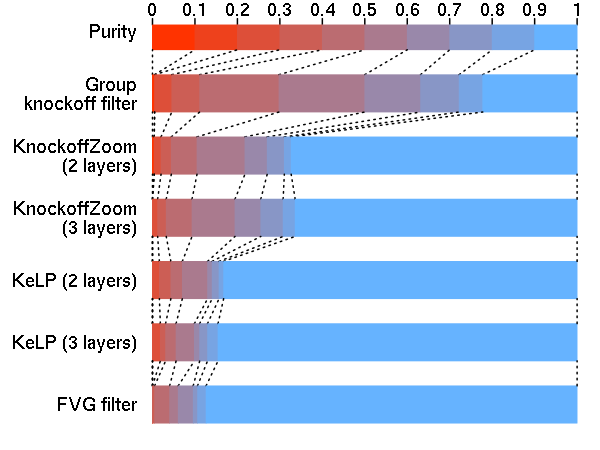}
			\par{(d) Comparison of catching set purity among different approaches.}
		\end{minipage}
		\caption{Comparison of our FVG filter and other methods in analyzing the EADB--UKBB dataset.}
		\label{fig:Real}
	\end{figure}

	Without power loss, our FVG filter {returns informative selection sets via simultaneous fine-mapping}. This can be seen in Figure \ref{fig:Real} (c)-(d) where empirical distributions of catching set sizes and purity are displayed. Specifically, all but one catching sets  of our FVG filter are of size no larger than $8$ while most of catching sets are of size 1 and purity greater than $90\%$. In other words, most of catching sets only contain the important variants or strongly correlated proxy variants. This results in the smallest average size (1.343) and highest average purity (0.951) of catching sets.  In contrast, 32 catching sets obtained by the group knockoff filter contain more than $8$ variants with strong impurity and are far from being informative. Specifically, for the five largest catching sets obtained by the group knockoff filter shown in Table \ref{Real_table2}, our FVG filter manages to distinguish one or two variants from more than 50 proxy variants in each. For comparison, KnockoffZoom (2 layers) only manages to narrow the signal of group ``6:31571218:32682663:64" down to the variant \texttt{rs532965} while the rest four groups are included as a whole in the resolution-adaptive discoveries. Adding the intermediate layer only refines the catching set of group  ``17:43056905:45876021:36" (size: 314) to a catching set of subgroup with size 281 and the catching set of group  ``6:31571218:32682663:83" (size: 57) to a catching set of subgroup with size 19. Thus, the average size and purity of catching sets provided by KnockoffZoom is still suboptimal. In addition, catching sets provided by our FVG filter also dominate the KeLP in size. Specifically, among all catching sets provided by our FVG filter, more than 95\% are if size 1 or 2 (less than 85\% for KeLP (2 layers) and KeLP (3 layers)). This suggests the advantages of the proposed FVG filter in refining AD-associated genetic discoveries over KnockoffZoom and KeLP. Comprehensive comparisons of our method with all existing methods are summarized in Table \ref{Tab:motivation2}.
	
	\renewcommand{\arraystretch}{0.9}
	\begin{table}
		\centering
		\caption{List of variants identified by the FVG filter  in the five largest catching sets obtained by the group knockoff filter.}\label{Real_table2}
		\resizebox{\columnwidth}{!}{
			\begin{tabular}{lrrrrrr}
				\toprule
				Identified&\multirow{2}{*}{Chromosome}&\multirow{2}{*}{Position}&Catching sets obtained by&Closest&\multirow{2}{*}{cS2Gene}&\multirow{2}{*}{Annotation}\\
				Variant&&&the group knockoff filter (size)&Gene&&\\
				\midrule
				rs3132451&6&31614248&6:31571218:32682663:7 (232)&AIF1&AIF1&Promoter,EpiMap\\\midrule
				rs532965&6&32610196&6:31571218:32682663:64 (84)&HLA-DRB1&HLA-DRB5&ABC\\\midrule
				rs9469112&6&32447376&6:31571218:32682663:76 (117)&HLA-DRA&&\\\midrule
				rs9271162&6&32609938&\multirow{2}{*}{6:31571218:32682663:83 (57)}&HLA-DRB1&HLA-DRB1&EpiMap\\
				rs9273349&6&32658092&&HLA-DQB1&&\\\midrule
				rs199451&17&46724418&17:43056905:45876021:36 (314)&NSF&KANSL1-AS1&GTeX\\
				\bottomrule
		\end{tabular}}
	\end{table}
	
	\renewcommand{\arraystretch}{0.9}
	\begin{table}
		\centering\caption{Summary of results by applying our FVG filter to the EADB-UKBB dataset in comparison with existing methods, where bold values indicate the best performance under different evaluation metrics.}\label{Tab:motivation2}
		\resizebox{\columnwidth}{!}{\begin{tabular}{lrrrr}
				\toprule
				\multirow{3}*{Method}&Number of &Average number of &Average&Average \\
				&identified&variants per & size of& purity of\\
				&loci&identified locus&catching sets&catching sets\\
				\midrule
				Marginal association test&54&17.093&17.093&0.489\\
				Model-X knockoff filter&38&{4.263}&1.000$^\dagger$&1.000$^\dagger$\\
				Group knockoff filter&\textbf{91}&19.923&9.157&0.651\\
				KnockoffZoom (2 layers)&\textbf{91}&15.152&5.621&0.851\\
				KnockoffZoom (3 layers)&\textbf{91}&14.681&5.409&0.858\\
				KeLP (2 layers)&45&10.422&2.535&0.918\\
				KeLP (3 layers)&38&10.605&2.488&0.932\\
				\midrule
				FVG filter&89&\textbf{2.989}&\textbf{1.343}&\textbf{0.951}\\
				\bottomrule
				\multicolumn{5}{l}{\small $^\dagger$: Size and purity of catching sets provided by model-X knockoff filter are trivially 1 and thus not included}\\
				\multicolumn{5}{l}{\small \hspace{0.2cm} in comparison.}\\
		\end{tabular}}
	\end{table}
	
	{In addition to statistical performance, the biological informativeness of discoveries made by our FVG filter is also our concern.  We apply a SNP-to-gene linking strategy \citep{Gazal2022} to investigate whether the identified genetic variants are functionally enriched or not. Here, genetic variants in the same group can functionally annotate different genes or have different annotation types. There are in total 7 types of annotations as follows.
		\begin{itemize}
			\item Exon: The genetic variant is located less than 20 base pairs from an exon of the annotated gene.
			\item Promoter: The genetic variant lies in promoter regions of the annotated gene.
			\item GTEx: The genetic variant lies in an expression quantitative trait locus of the annotated gene across 54 cell types, according to the GTEx data.
			\item eQTLGen: The genetic variant lies in an expression quantitative trait locus of the annotated gene in blood, according to the eQTLGen data.
			\item Epimap: The genetic variant lies in enhancer regions of the annotated gene across 833 cell types under the EpiMap model.
			\item ABC: The genetic variant lies in Hi-C linked enhancer regions of the annotated gene across 167 cell types under the Activity-by-Contact model.
			\item Cicero: The genetic variant has co-accessibility links to the annotated gene's promoter across 61,806 blood/basal cells.
		\end{itemize}
		Among 266 genetic variants identified by our FVG filter in Table \ref{Real_table1}, 180 (67.67\%) can be mapped with function evidence. This proportion of functional annotation is significantly higher than the average percentage of the background genome (30.52\%) and variants in catching sets of the group knockoff filter (56.54\%) as shown in Table \ref{Real_table4}. Such an advantage is more significant in annotation types ``Exon", ``Promoter" and ``ABC". As a result, compared to the group-level inference, discoveries made by the FVG filter is more biological informative for further drug development or therapy designs.}

	\renewcommand{\arraystretch}{1}
	\begin{table}
		\centering
		\caption{Proportions of different annotations for variants in catching sets of group knockoff filter and the FVG filter.}\label{Real_table4}
		\resizebox{\columnwidth}{!}{\begin{tabular}{lrrr}
				\toprule
				Annotation&Background Genome&Group knockoff filter&FVG filter\\\midrule
				Exon&6.74\%&11.64\%&\textbf{19.55\%}\\
				Promoter&2.68\%&7.67\%&\textbf{12.03\%}\\
				GTeX&8.33\%&22.78\%&\textbf{24.44\%}\\
				eQTLGen&1.57\%&\textbf{6.56\%}&6.02\%\\
				Epimap&10.10\%&12.24\%&\textbf{12.41\%}\\
				ABC&8.87\%&9.76\%&\textbf{14.66\%}\\
				Cicero&0.98\%&3.09\%&\textbf{4.51\%}\\\midrule
				Total&30.52\%&56.54\%&\textbf{67.67\%}\\
				\bottomrule
		\end{tabular}}
	\end{table}
	
	\section{Discussions}\label{Discussions}

	As more genetic variants are sequenced using the rapidly developing whole-genome sequencing technology, it is now possible to explain more diseases' heritability via large-scale genetic studies. Noting that existing knockoff methods either lose power or lack informativeness because of strong correlations among variants in Alzheimer's disease analysis of the EADB-UKBB dataset, we develop a new filter {that uses group knockoffs to simultaneously select important groups and perform fine-mapping within each selected group}. Specifically, we first define a new family of conditional independence hypotheses that allow inference at feature level under group knockoffs construction. Based on theoretical properties of test statistics when group knockoffs are used, we develop the {FVG} knockoff filter with  FDR control. Utilizing penalized regressions, our approach can efficiently learn the sparse feature importance and refine the catching set within each group by selecting a subset of most promising features.  Extensive simulated experiments with real-world genetic data empirically validate the FDR control of our proposed filter. Compared with the existing knockoff filters, our FVG filter is shown {to return small and pure catching sets without missing many important features. When applying the proposed method to AD analysis of the real-world EADB--UKBB dataset, our FVG filter achieves the balance between power (identifies as many signals as the most powerful group knockoff filter) and precision of catching sets (with small average size and high average purity).}
	
	Simultaneous inference of conditional independence at different layers of groups is of great necessity as we need to balance the power and the resolution (informativeness) of groups \citep{Katsevich2019,Sesia2020,Gablenz2024}. For example, in genetic analysis, variant groups with higher resolution are typically smaller and have higher correlations, making it challenging to identify important signals, but easier to interpret if they are identified. Thus, it is interesting to incorporate the proposed filter into the multilayer testing framework \citep{Barber2017,Katsevich2019,Sesia2020,Gablenz2024} to perform simultaneous inference of $H^{(\text{fg})}$'s at multiple layers of group structure. In addition, it is also of great interest to adopt the proposed filter in causal inference. By elaborately designing group knockoffs that can eliminate confounding effects, the proposed filter can perform causal features selection with provable FDR control when unmeasured confounders exist. {We leave these tasks as future works}.
	
	 \begin{appendix}
	 	
	 	\section{Supplementary Figures for Section \ref{Existing}}\label{Supp_fig}
	 	
	 	\begin{figure}[H]
	 		\centering
	 		\begin{minipage}[t]{0.32\linewidth}
	 			\centering
	 			\includegraphics[width=\linewidth]{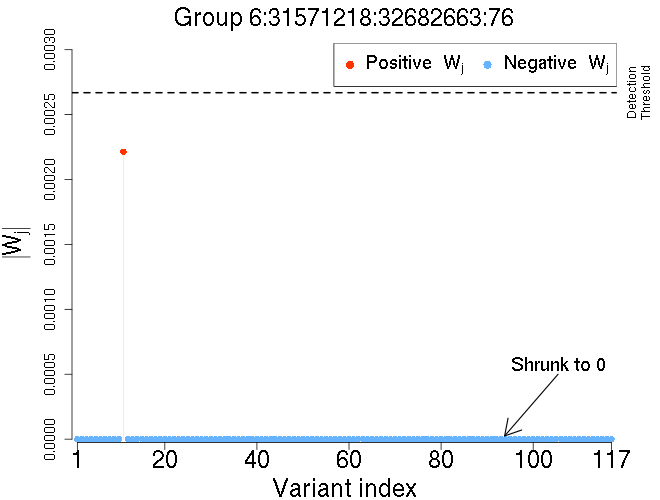}
	 			\par{(a) Manhattan plot of feature importance statistics of different genetic variants under KnockoffZoom \\(2 layers) at the feature-versus-feature level.}
	 		\end{minipage}
	 		\begin{minipage}[t]{0.32\linewidth}
	 			\centering
	 			\includegraphics[width=\linewidth]{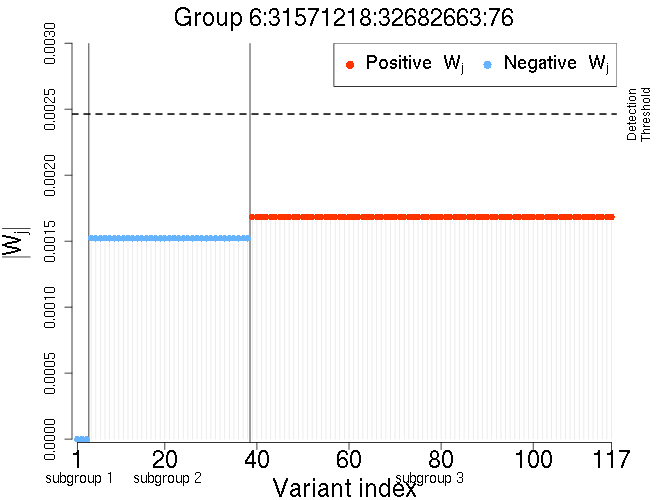}
	 			\par{(b) Manhattan plot of the contribution of different variants to the importance statistics of the group under  KnockoffZoom (3 layers) at the intermediate layer.}
	 		\end{minipage}
	 		\begin{minipage}[t]{0.32\linewidth}
	 			\centering
	 			\includegraphics[width=\linewidth]{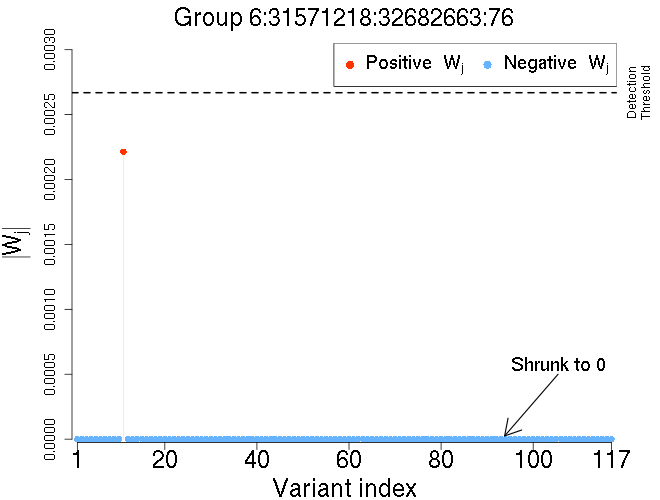}
	 			\par{(c) Manhattan plot of feature importance statistics of different genetic variants under KnockoffZoom \\(3 layers) at the feature-versus-feature level.}
	 		\end{minipage}
	 		\caption{Manhanttan plot under KnockoffZoom (2 layers) and KnockoffZoom (3 layers) in a high-LD region between positions 31571218 $\sim$ 32682663,  chromosome 6. Here we exclude plots at the group-versus-group level that is the same as Figure \ref{fig:motivating_cluster} (c).}
	 		\label{fig:motivating_cluster2}
	 	\end{figure}
	 	
	 	\section{Proof of Theorem \ref{thm:uniform_two}}\label{pr:uniform_two}
	 	Without loss of generality, we let $k=1$, $k^\dagger=2$, $j\in B_1$ and $j^\dagger\in B_2$. 
	 	
	 	\begin{itemize}
	 		\item[$\star$] {\rm \textbf{(Uniformity)}}:\\ 
	 		Since {\rm$H_j^{(\text{fg})}$} implies {\rm$H_k^{\text{(gg)}}$} where $j\in B_k$, if $j\in\mathcal{H}_0^{(\text{fg})}$, we have $H_1^{\text{(gg)}}$ is true (and $H_{j'}^{(\text{fg})}$ is true for all $j'\in B_k$) and $(X_{j},\textbf{X}_{B_1\setminus\{j\}})\perp Y|\textbf{X}_{-B_1}$ by (\ref{H_gg}). Because 
	 		$(\widetilde{X}_{j},\widetilde{\textbf{X}}_{B_1\setminus\{j\}},\widetilde{\textbf{X}}_{-B_1})\perp Y|(X_{j},\textbf{X}_{B_1\setminus\{j\}},\textbf{X}_{-B_1})$, we have
	 		$$(X_{j},\textbf{X}_{B_1\setminus\{j\}},\widetilde{X}_{j},\widetilde{\textbf{X}}_{B_1\setminus\{j\}},\widetilde{\textbf{X}}_{-B_1})\perp Y|\textbf{X}_{-B_1},$$
	 		and thus
	 		\begin{equation}
	 			\label{cond_single_two}
	 			(X_{j},\textbf{X}_{B_1\setminus\{j\}},\widetilde{X}_{j},\widetilde{\textbf{X}}_{B_1\setminus\{j\}})\perp Y|(\textbf{X}_{-B_1},\widetilde{\textbf{X}}_{-B_1}).
	 		\end{equation}
	 		
	 		By \eqref{group_exchangeability}, we have conditional on $(\textbf{X}_{-B_1},\widetilde{\textbf{X}}_{-B_1})$, 
	 		\begin{equation}\label{conditional_exchangeability}
	 			(X_{j},\textbf{X}_{B_1\setminus\{j\}},\widetilde{X}_{j},\widetilde{\textbf{X}}_{B_1\setminus\{j\}}){\displaystyle\mathop{=\joinrel=}^{D}}(\widetilde{X}_{j},\widetilde{\textbf{X}}_{B_1\setminus\{j\}},X_{j},\textbf{X}_{B_1\setminus\{j\}}).
	 		\end{equation}
	 		Thus, \eqref{cond_single_two}-\eqref{conditional_exchangeability} lead to 
	 		\begin{equation}\label{Equal_in_dist}
	 			\begin{aligned}[b]
	 				&([\mathbb{X}_j,\mathbb{X}_{B_1\setminus\{j\}},\widetilde{\mathbb{X}}_j,\widetilde{\mathbb{X}}_{B_1\setminus\{j\}},\{\mathbb{X}_{-B_1},\widetilde{\mathbb{X}}_{-B_1}\}],\textbf{y})\\
	 				{\displaystyle\mathop{=\joinrel=}^{D}}&([\widetilde{\mathbb{X}}_j,\widetilde{\mathbb{X}}_{B_1\setminus\{j\}},{\mathbb{X}}_j,{\mathbb{X}}_{B_1\setminus\{j\}},\{\mathbb{X}_{-B_1},\widetilde{\mathbb{X}}_{-B_1}\}],\textbf{y}).
	 			\end{aligned}
	 		\end{equation}
	 		
	 		Note that swapping $\mathbb{X}_j,\mathbb{X}_{B_1\setminus\{j\}}$ and $\widetilde{\mathbb{X}}_j,\widetilde{\mathbb{X}}_{B_1\setminus\{j\}}$ does not change $|W_1|,\ldots,|W_p|$ and $\text{\rm sign}(W_{j})$ for those false $H^{(\text{fg})}_{j}$'s, \eqref{ig_type_two} and \eqref{Equal_in_dist} implies that 
	 		conditional on $|W_1|,\ldots,|W_p|$ and $\text{\rm sign}(W_{j})$ for those false $H^{(\text{fg})}_{j}$'s,
	 		\begin{equation}\label{Uniform_swap2}
	 			(T_j,\widetilde{T}_j){\displaystyle\mathop{=\joinrel=}^{D}}(\widetilde{T}_j,T_j),
	 		\end{equation}
	 		and thus $\text{\rm sign}(W_{j})$ uniformly distributes on $\{+,-\}$.
	 		
	 		%
	 		\item[$\star$] {\rm \textbf{(Between-group Independence)}}:\\ By \eqref{group_exchangeability}, we have 
	 		{$$
	 			\begin{aligned}
	 				&({\textbf{X}}_{ B_1},\widetilde{\textbf{X}}_{ B_1},{\textbf{X}}_{ B_2},\widetilde{\textbf{X}}_{ B_2})|(\textbf{X}_{-(B_1\cup B_2)},\widetilde{\textbf{X}}_{-(B_1\cup B_2)})
	 				\\
	 				{\displaystyle\mathop{=\joinrel=}^{D}}
	 				&({\textbf{X}}_{ B_1},\widetilde{\textbf{X}}_{ B_1},\widetilde{\textbf{X}}_{ B_2},{\textbf{X}}_{ B_2})|(\textbf{X}_{-(B_1\cup B_2)},\widetilde{\textbf{X}}_{-(B_1\cup B_2)})
	 				\\
	 				{\displaystyle\mathop{=\joinrel=}^{D}}
	 				&(\widetilde{\textbf{X}}_{ B_1},{\textbf{X}}_{ B_1},{\textbf{X}}_{ B_2},\widetilde{\textbf{X}}_{ B_2})|(\textbf{X}_{-(B_1\cup B_2)},\widetilde{\textbf{X}}_{-(B_1\cup B_2)})
	 				\\
	 				{\displaystyle\mathop{=\joinrel=}^{D}}
	 				&(\widetilde{\textbf{X}}_{ B_1},{\textbf{X}}_{ B_1},\widetilde{\textbf{X}}_{ B_2},{\textbf{X}}_{ B_2})|(\textbf{X}_{-(B_1\cup B_2)},\widetilde{\textbf{X}}_{-(B_1\cup B_2)})
	 			\end{aligned}$$}
	 		If $j,j^\dagger\in\mathcal{H}_0^{(\text{fg})}$, by (\ref{ig_type_two}) and (\ref{Uniform_swap2}), we have conditional on $|W_1|,\ldots,|W_p|$ and $\text{\rm sign}(W_{j})$ for those false $H^{(\text{fg})}_{j}$'s,
	 		$$({T}_{j},\widetilde{T}_{j},{T}_{j^\dagger},\widetilde{T}_{j^\dagger}){\displaystyle\mathop{=\joinrel=}^{D}}({T}_{j},\widetilde{T}_{j},\widetilde{T}_{j^\dagger},{T}_{j^\dagger})
	 		{\displaystyle\mathop{=\joinrel=}^{D}}(\widetilde{T}_{j},{T}_{j},{T}_{j^\dagger},\widetilde{T}_{j^\dagger})
	 		{\displaystyle\mathop{=\joinrel=}^{D}}(\widetilde{T}_{j},{T}_{j},\widetilde{T}_{j^\dagger},{T}_{j^\dagger}).$$
	 		and thus $\text{\rm sign}(W_{j})$ and $\text{\rm sign}(W_{j^\dagger})$ are independent conditional on $|W_1|,\ldots,|W_p|$ and $\text{\rm sign}(W_{j})$ for those false $H^{(\text{fg})}_{j}$'s.
	 	\end{itemize}

	 	\section{Proof of Theorem \ref{thm:FDR}}\label{pr:FDR}
	 	
	 	As feature statistics $W_{j}$'s satisfy the between-group coin-flipping property, we have that $W_j$'s within each row of Table \ref{tab:align} satisfy the  coin-flipping property introduced in \citet{Candes2018}. In other words, for each $l$, $\text{sign}(W_{j})$'s for those $j\in \mathcal{C}_l\cap \mathcal{H}_0^{(\text{fg})}$ are i.i.d. coin flips conditional on $\{W_j|j\in \mathcal{C}_l\}$ and $\{\text{\rm sign}(W_{j})|j\in \mathcal{C}_l\setminus \mathcal{H}_0^{(\text{fg})}\}$. Thus, by Lemma 3 of \citet{Katsevich2019}, we have
	 	\begin{equation}\label{E_sup}
	 		\begin{aligned}[b]
	 			\textbf{E}\left\{\sup_{t>0}\frac{V^{(l)}_{\mathcal{H}_0,+}(t)}{1+V^{(l)}_{\mathcal{H}_0,-}(t)}\right\}\leq 1.93,\quad \text{where }\begin{cases}
	 				V^{(l)}_{\mathcal{H}_0,+}(t)=\#\{j\in \mathcal{C}_l\cap \mathcal{H}_0^{(\text{fg})}|W_{j}\geq t\},\\
	 				V^{(l)}_{\mathcal{H}_0,-}(t)=\#\{j\in \mathcal{C}_l\cap \mathcal{H}_0^{(\text{fg})}|W_{j}\leq -t\}.\\
	 			\end{cases}
	 		\end{aligned}
	 	\end{equation}
	 	Thus, the FDR of the rejection set {\rm $\mathcal{R}^{(\text{fg})}$} obtained by solving \eqref{opt} satisfies 
	 	\begin{align}
	 		{\text{FDR}}^{(\text{fg})}=&\mathbf{E}\Biggl\{\frac{\#(\mathcal{R}^{(\text{fg})}\cap \mathcal{H}_0^{(\text{fg})})}{1\vee\#\mathcal{R}^{(\text{fg})}}\Biggr\}\nonumber\\
	 		=&\mathbf{E}\Biggl\{\frac{\sum_{l=1}^{\infty} \text{I}\left(\max_{j\in \mathcal{C}_l}|W_{j}|\geq t^{(l)}\right)\cdot V^{(l)}_{\mathcal{H}_0,+}(t^{(l)})}{1\vee(\sum_{l}\{j\in \mathcal{C}_l|W_{j}\geq t^{(l)}\})}\Biggr\}\nonumber\\
	 		=&\mathbf{E}\Biggl\{\sum_{l=1}^{\infty}\text{I}\left(\max_{j\in \mathcal{C}_l}|W_{j}|\geq t^{(l)}\right)\cdot\frac{ 1+V^{(l)}_{\mathcal{H}_0,-}(t^{(l)})}{1\vee(\sum_{l}\{j\in \mathcal{C}_l|W_{j}\geq t^{(l)}\})}\cdot\frac{V^{(l)}_{\mathcal{H}_0,+}(t^{(l)})}{1+V^{(l)}_{\mathcal{H}_0,-}(t^{(l)})}\Biggr\}\nonumber\\
	 		\leq&\mathbf{E}\Biggl\{\sum_{l=1}^{\infty}\text{I}\left(\max_{j\in \mathcal{C}_l}|W_{j}|\geq t^{(l)}\right)\cdot\frac{1+\#\{j\in \mathcal{C}_l|W_{j}\leq -t^{(l)} \}}{1\vee[\sum_{l}\#\{j\in \mathcal{C}_l|W_{j}\geq t^{(l)} \}]}\cdot\frac{V^{(l)}_{\mathcal{H}_0,+}(t^{(l)})}{1+V^{(l)}_{\mathcal{H}_0,-}(t^{(l)})}\Biggr\}\nonumber\\
	 		\leq&\mathbf{E}\Biggl\{\sum_{l=1}^{\infty}\frac{v_l\alpha}{1.93}\cdot\frac{V^{(l)}_{\mathcal{H}_0,+}(t^{(l)})}{1+V^{(l)}_{\mathcal{H}_0,-}(t^{(l)})}\Biggr\}\nonumber\\
	 		\leq&\mathbf{E}\Biggl\{\sum_{l=1}^{\infty}\frac{v_l\alpha}{1.93}\cdot\sup_{t>0}\frac{V^{(l)}_{\mathcal{H}_0,+}(t)}{1+V^{(l)}_{\mathcal{H}_0,-}(t)}\Biggr\}\nonumber\\
	 		=&\sum_{l=1}^{\infty}\frac{v_l\alpha}{1.93}\cdot\mathbf{E}\Biggl\{\sup_{t>0}\frac{V^{(l)}_{\mathcal{H}_0,+}(t)}{1+V^{(l)}_{\mathcal{H}_0,-}(t)}\Biggr\}\nonumber\\
	 		=&\sum_{l=1}^{\infty}{v_l\alpha}\nonumber\\
	 		=&\alpha,\label{proof_of_FDR}
	 	\end{align}
	 	by \eqref{E_sup}.
	 	
	 	{\section{Comparisons between the Naive Filter and the FVG Filter}\label{Comparison2}
	 		In this section, we empirically compare the naive filter (Algorithm \ref{alg:G-FDR}), the FVG filter (Algorithm \ref{alg:G-FDR2}) with and without the $1.93$ factor using simulated datasets under the setting of Section \ref{Sim_setting}. Over 1000 simulated datasets of different sample sizes, the empirical FDR and power under different target FDR levels are reported in Figure \ref{fig:sample_size3}. Although both the naive filter and the FVG filter without the 1.93 factor does not possess theoretically rigorous FDR control, their empirical FDRs remain under control. In addition, they both have higher power than the FVG filter (Algorithm \ref{alg:G-FDR2}) with the $1.93$ factor.}
	 	
	 	\begin{figure}[t]
	 		\centering
	 		\begin{minipage}{1\linewidth}
	 			\centering
	 			\includegraphics[width=\linewidth]{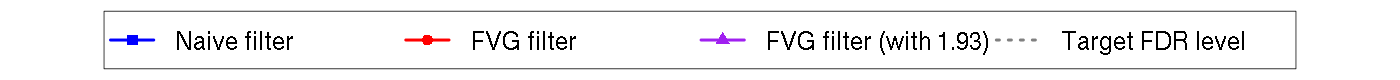}
	 		\end{minipage}\\
	 		\begin{minipage}{0.32\linewidth}
	 			\centering
	 			\includegraphics[width=\linewidth]{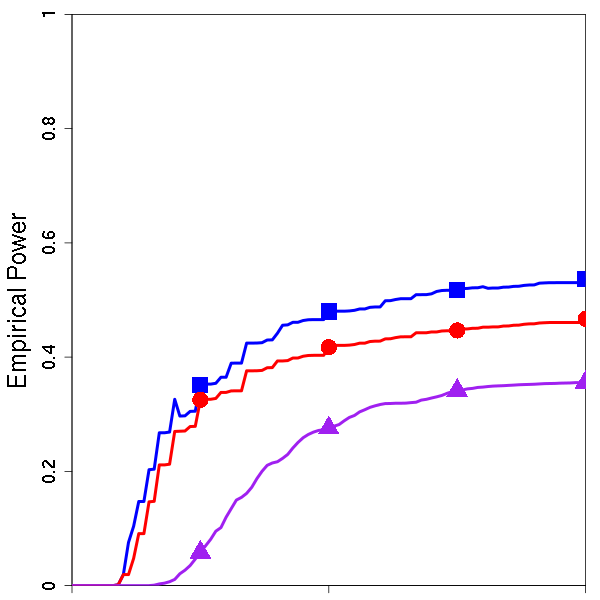}
	 			\includegraphics[width=\linewidth]{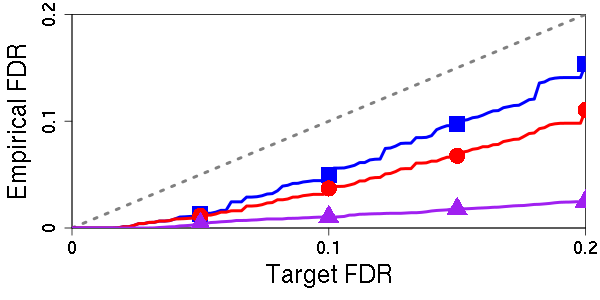}
	 			\par{(a) $n=500$.}
	 		\end{minipage}
	 		\begin{minipage}{0.32\linewidth}
	 			\centering
	 			\includegraphics[width=\linewidth]{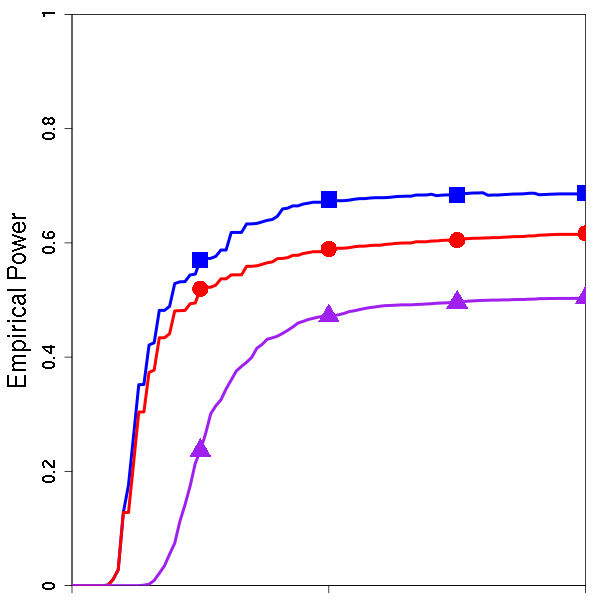}
	 			\includegraphics[width=\linewidth]{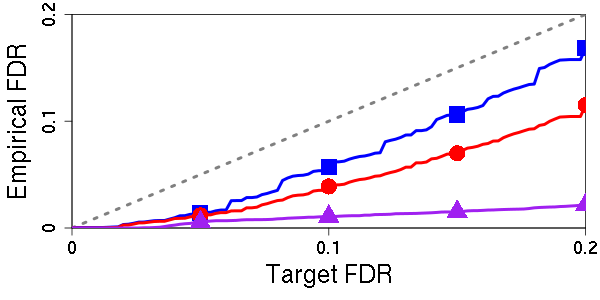}
	 			\par{(b) $n=1000$.}
	 		\end{minipage}
	 		\begin{minipage}{0.32\linewidth}
	 			\centering
	 			\includegraphics[width=\linewidth]{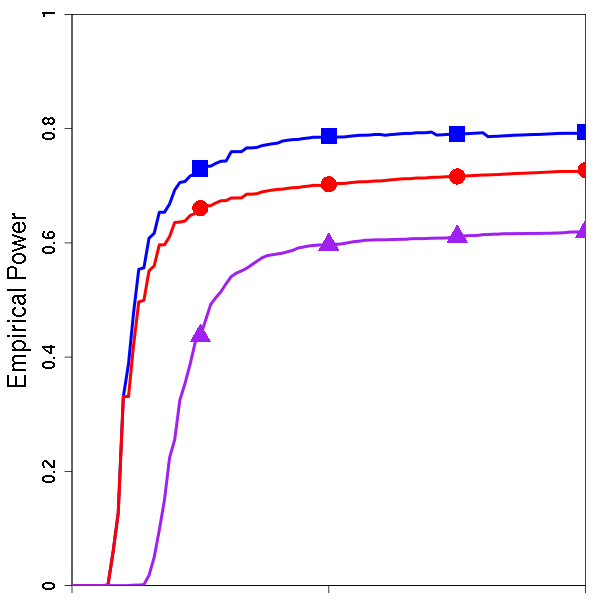}
	 			\includegraphics[width=\linewidth]{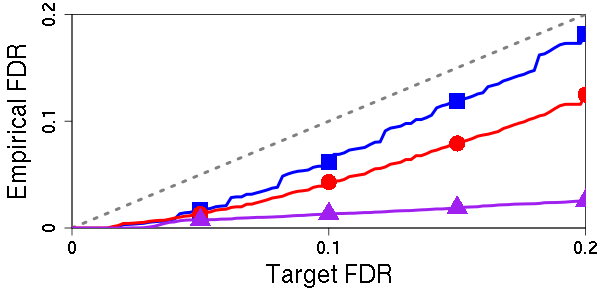}
	 			\par{(c) $n=2000$.}
	 		\end{minipage}
	 		\caption{Empirical FDR and power of the naive filter (Algorithm \ref{alg:G-FDR}), the FVG filter (Algorithm \ref{alg:G-FDR2}) with and without the $1.93$ factor with respect to the target FDR level ($\alpha$) over 1000 simulated genetic datasets of different sample sizes.}
	 		\label{fig:sample_size3}
	 	\end{figure} 
	 	
	 	\section{Extensions}\label{extensions}
	 	
	 	\subsection{Multiple Knockoffs}\label{multiple}
	 	
	 	Although Algorithm \ref{alg:G-FDR2} with group knockoffs can obtain the rejection set with FDR control, randomness in group knockoffs generation could produce greatly different {\rm $\mathcal{R}^{(\text{fg})}$} in different runs, especially in the case that the number of false $H_j^{(\text{fg})}$'s is small  and Algorithm \ref{alg:G-FDR2} could return an empty rejection set \citep{Gimenez2019b}. 
	 	Here we extend the proposed Algorithm \ref{alg:G-FDR2} to multiple group knockoffs $\widetilde{\textbf{X}}^{(1)},\ldots,\widetilde{\textbf{X}}^{(M)}$ where original and knockoff features are simultaneously exchangeable at the level of feature groups. That is to say, with the convention that $\textbf{X}=\widetilde{\textbf{X}}^{(0)}$,
	 	$\widetilde{\textbf{X}}^{(0)},\widetilde{\textbf{X}}^{(1)},\ldots,\widetilde{\textbf{X}}^{(M)}$  satisfy
	 	{\rm \begin{equation}\label{group_exchangeability_multiple}
	 			\begin{aligned}[b]
	 				&(\widetilde{\textbf{X}}^{(\sigma_1(0))}_{B_1},\ldots,\widetilde{\textbf{X}}^{(\sigma_K(0))}_{B_K},\widetilde{\textbf{X}}^{(\sigma_1(1))}_{B_1},\ldots,\widetilde{\textbf{X}}^{(\sigma_K(1))}_{B_K},\ldots,\widetilde{\textbf{X}}^{(\sigma_1(M))}_{B_1},\ldots,\widetilde{\textbf{X}}^{(\sigma_K(M))}_{B_K})\\
	 				{\displaystyle\mathop{=\joinrel=}^{D}}&
	 				(\widetilde{\textbf{X}}^{(0)}_{B_1},\ldots,\widetilde{\textbf{X}}^{(0)}_{B_K},\widetilde{\textbf{X}}^{(1)}_{B_1},\ldots,\widetilde{\textbf{X}}^{(1)}_{B_K},\ldots,\widetilde{\textbf{X}}^{(M)}_{B_1},\ldots,\widetilde{\textbf{X}}^{(M)}_{B_K})
	 			\end{aligned}
	 	\end{equation}}for any permutations $\sigma_1,\ldots,\sigma_K$ {of }$\{0,1,\ldots,M\}$. With importance scores $\{T_j^{(m)}|j=1,\ldots,p;m=0,\ldots,M\}$ obtained in an analogous way of (\ref{ig_type_two}), we follow \citet{He2021} to compute feature statistics 
	 	\begin{equation}\label{def:inf_stat}
	 		\kappa_{j}=\arg\max_m T_{j}^{(m)},\quad\tau_{j}=\max_m T_{j}^{(m)}-\text{\rm median}\{T_{j}^{(m)}|m\neq \kappa_{j}\},\quad j=1,\ldots,p.
	 	\end{equation}
	 	Specifically, $\kappa_{j}$ and $\tau_{j}$ are the multiple group knockoff counterparts of $\text{sign}(W_j)$ and $|W_j|$ respectively \citep{Gimenez2019b} that 
	 	\begin{enumerate}
	 		\item[$\star$] conditional on $\tau_1,\ldots,\tau_p$ and $\kappa_{j}$'s for those false $H^{(\text{fg})}_{j}$'s,
	 		\begin{enumerate}
	 			\item[$\diamond$] {\rm \textbf{(Uniformity)}} $\kappa_{j}$ uniformly distributes on $\{0,1,\ldots,M\}$  for all {\rm$j\in\mathcal{H}_0^{(\text{fg})}$};
	 			\item[$\diamond$] {\rm \textbf{(Between-Group Independence)}}
	 			for any $k\neq k^\dagger$, {\rm$\kappa_{j}$} and {\rm$\kappa_{j^\dagger}$} are independent for any $j\in B_k\cap \mathcal{H}_0^{(\text{fg})}$ and $j^\dagger\in B_{k^\dagger}\cap \mathcal{H}_0^{(\text{fg})}$.
	 		\end{enumerate}
	 	\end{enumerate}
	 	Based on $\kappa_{j}$'s and $\tau_{j}$'s, we provide Algorithm \ref{alg:G-FDR_multiple} as the multiple knockoff counterpart of  Algorithm \ref{alg:G-FDR2}.
	 	
	 	
	 	{\begin{algorithm}
	 			\caption{Feature filter with multiple group knockoffs (naive version).}\label{alg:G-FDR_multiple}
	 			\begin{algorithmic}[1]
	 				\STATE  \textbf{Input:} 
	 				Groups $B_1,\ldots,B_K$, feature statistics $\{(\kappa_{j},\tau_{j})|j=1,\ldots,p\}$  and the target level $\alpha>0$.
	 				\STATE Align $(\kappa_{j},\tau_{j})$'s as shown in Table \ref{tab:align2} such that $(\kappa_{j},\tau_{j})$'s for all features in the group $B_k$ are in the $k$-th column and $\tau_{(k1)}\geq \tau_{(k2)}\geq \cdots$.
	 				\STATE Compute budget $v_l$ for the $l$-th row using $\tau_1,\ldots,\tau_p$ for $l=1,2,\ldots$ such that $\sum_{l=1}^{\infty}v_l=1$.
	 				\STATE Compute grids $\mathcal{G}_l=\{1/v_l,2/v_l,\ldots,({neg}_l+1)/v_l\}$ for the $l$-th row where
	 				$neg_l=\#\{j\in \mathcal{C}_l|\kappa_{j}\neq 0\}$.
	 				\STATE Combine $\mathcal{G}_{comb}=(\cup_l \mathcal{G}_l)\cup \{0\}$ and sort values ${grid}_{(1)}> {grid}_{(2)}>\ldots$ in $\mathcal{G}_{comb}$. 
	 				\STATE Initialize $b=0$
	 				\REPEAT
	 				\STATE Update $b=b+1$.
	 				\STATE Compute 
	 				{\begin{equation}\label{t_alpha3}
	 						t^{(l)}=\min\Biggl\{t>0\Bigg|\frac{1+\#\{j\in \mathcal{C}_l|\kappa_{j}\neq 0,\tau_{j}\geq t\}}{v_l}\leq {grid}_{(b)}\Biggr\},
	 				\end{equation}}
	 				for $l=1,2,\ldots$.
	 				\UNTIL{$\frac{\text{I}\left(\max_{j\in \mathcal{C}_l}|W_{j}|\geq t\right)}{M}\times\frac{1+\#\{j\in \mathcal{C}_l|\kappa_j\neq 0,\tau_{j}\geq t^{(l)} \}}{1\vee[\sum_{l}\#\{j\in \mathcal{C}_l|\kappa_j= 0,\tau_{j}\geq t^{(l)} \}]}\leq \frac{v_l\alpha}{1.93}$ for all $l$.}
	 				\STATE \textbf{Output:} The rejection set $\mathcal{R}^{(\text{fg})}=\cup_{l}\{j\in  \mathcal{C}_l|\kappa_{j}= 0,\tau_{j}\geq t^{(l)}\}$.
	 			\end{algorithmic}
	 	\end{algorithm}}

	 	\renewcommand{\arraystretch}{0.9}
	 	\begin{table}
	 		\centering
	 		\caption{Alignment of feature statistics $(\kappa_{j},\tau_{j})$'s such that different groups correspond to different columns. Here, $(\kappa_{(kl)},\tau_{(kl)})$ corresponds to the feature from the $k$-th group that is aligned in the $l$-th row.}
	 		\label{tab:align2}
	 		\begin{tabular}{ccccc}
	 			\toprule
	 			Row&$B_1$&$B_2$&$B_3$&$\cdots$\\
	 			\midrule
	 			1&$(\kappa_{(11)},\tau_{(11)})$&$(\kappa_{(21)},\tau_{(21)})$&$(\kappa_{(31)},\tau_{(31)})$&$\cdots$\\
	 			\hdashline
	 			2&$(\kappa_{(12)},\tau_{(12)})$&$(\kappa_{(22)},\tau_{(22)})$&$(\kappa_{(32)},\tau_{(32)})$&$\cdots$\\
	 			\hdashline
	 			3&$(\kappa_{(13)},\tau_{(13)})$&$(\kappa_{(23)},\tau_{(23)})$&$(\kappa_{(33)},\tau_{(33)})$&$\cdots$\\
	 			\hdashline
	 			4&$(\kappa_{(14)},\tau_{(14)})$&$(\kappa_{(24)},\tau_{(24)})$&$(\kappa_{(34)},\tau_{(34)})$&$\cdots$\\
	 			\hdashline
	 			$\vdots$&$\vdots$&$\vdots$&$\vdots$&$\ddots$\\
	 			\bottomrule
	 		\end{tabular}
	 	\end{table}

	 	\subsection{Filter based on E-values}
	 	
	 	Another variation relies on the connection between the knockoff filter and $e$-values discussed in \citet{Ren2024}. Following \citet{Ren2024}'s way of computing $e$-values of a set of hypotheses whose $W_j$'s possess the coin-flipping property, we can analogously compute $e$-values of hypotheses whose $W_j$'s are aligned in the same row of Table \ref{tab:align} as shown in steps 4-7 of Algorithm \ref{alg:G-FDR4}. By doing so, an e-value $e_j$ is calculated for each hypothesis $H_{j}^{(\text{fg})}$ and we obtain the rejection set $\mathcal{R}^{(\text{fg})}$ by applying the e-BH procedure \citep{Wang2022} to $e$-values $e_1,\ldots,e_p$ (as shown in steps 8-10 of Algorithm \ref{alg:G-FDR4}).

	 	\begin{algorithm}
	 		\caption{Rigorous feature filter with group knockoffs using e-values.}\label{alg:G-FDR4}
	 		\begin{algorithmic}[1]
	 			\STATE  \textbf{Input:} Groups $B_1,\ldots,B_K$, feature statistics $W_1,\ldots,W_p$ and the target level $\alpha>0$.
	 			\STATE Align $W_j$'s as shown in Table \ref{tab:align} such that $W_j$'s for all features in the group $B_k$ are in the $k$-th column and $|W_{(k1)}|\geq |W_{(k2)}|\geq \cdots$.
	 			\STATE Compute the number of feature $p_l$ in the $l$-th row ($l=1,2,\ldots$).
	 			\FOR{$l=1,2,\ldots$}
	 			\STATE Compute
	 			$$
	 			t_{l,\alpha}=\min\Biggl\{t_l>0\Bigg|\frac{1+\#\{j \in \mathcal{C}_l|W_{j}\leq -t_l\}}{1\vee\#\{j \in \mathcal{C}_l|W_{j}\geq t_l\}}\leq \alpha/2\Biggr\}.
	 			$$
	 			\STATE Compute e-value $e_j$ for each $j\in \mathcal{C}_l$ as 
	 			$$
	 			e_j=\frac{p_l\times I(W_j\geq t_{l,\alpha})}{1+\#\{j \in \mathcal{C}_l|W_{j}\leq -t_{l,\alpha}\}}.
	 			$$
	 			\ENDFOR
	 			\STATE Compute $e_{(1)}\geq\ldots\geq e_{(p)}$ as the order statistics of $e_1,\ldots,e_p$.
	 			\STATE Compute $\widehat{J}=\max\{J|e_{(J)}\geq p/(\alpha J)\}$ or $\widehat{J}=0$ if $e_{(J)}< p/(\alpha J)$ for $J=1,\ldots,p$.
	 			\STATE \textbf{Output:} The rejection set $\mathcal{R}^{(\text{fg})}=\{j|e_j\geq e_{(\widehat{J})}\}$.
	 		\end{algorithmic}
	 	\end{algorithm}
	 	
	 	\begin{thm}\label{ThmS1}
	 		The rejection set {\rm $\mathcal{R}^{(\text{fg})}$} obtained by Algorithm \ref{alg:G-FDR4} controls {\rm $\text{FDR}^{(\text{fg})}$} at the target level  $\alpha>0$.
	 	\end{thm}
	 	
	 	\begin{proof}[Proof of Theorem \ref{ThmS1}]
	 		{\rm
	 			Because $W_j$'s that are in the same row of Table \ref{tab:align} possess the  coin-flipping property, by equation (7) of \citet{Ren2024}, we have
	 			$$\textbf{E}\left\{\frac{\sum_{j\in \mathcal{C}_l\cap \mathcal{H}_0^{(\text{fg})}} I(W_j\geq t_{l,\alpha})}{1+\sum_{j\in \mathcal{C}_l\cap \mathcal{H}_0^{(\text{fg})}} I(W_j\leq -t_{l,\alpha})}\right\}\leq 1.$$ As a result, we have for $l=1,2,\ldots,$
	 			\begin{align*}
	 				\sum_{j\in \mathcal{C}_l\cap \mathcal{H}_0^{(\text{fg})}}\textbf{E}\left\{e_j\right\}=&p_l\times \textbf{E}\left\{\frac{\sum_{j\in \mathcal{C}_l\cap \mathcal{H}_0^{(\text{fg})}} I(W_j\geq t_{l,\alpha})}{1+\sum_{j\in \mathcal{C}_l} I(W_j\leq -t_{l,\alpha})}\right\}\\\leq&
	 				p_l\times \textbf{E}\left\{\frac{\sum_{j\in \mathcal{C}_l\cap \mathcal{H}_0^{(\text{fg})}} I(W_j\geq t_{l,\alpha})}{1+\sum_{j\in \mathcal{C}_l\cap \mathcal{H}_0^{(\text{fg})}} I(W_j\leq -t_{l,\alpha})}\right\}\\
	 				\leq&
	 				p_l,
	 			\end{align*}
	 			and thus 
	 			\begin{align*}\sum_{j\in\mathcal{H}_0^{(\text{fg})}}\textbf{E}\left\{e_j\right\}=\sum_{l=1}^{\infty}\sum_{j\in \mathcal{C}_l\cap \mathcal{H}_0^{(\text{fg})}}\textbf{E}\left\{e_j\right\}
	 				\leq \sum_{l=1}^{\infty}p_l
	 				=p.\end{align*}
	 			Finally, by Theorem 2 of of \citet{Ren2024}, as $\sum_{j\in\mathcal{H}_0^{(\text{fg})}}\textbf{E}\left\{e_j\right\}\leq p$, the e-BH procedure in steps 8-10 of Algorithm \ref{alg:G-FDR4} returns the rejection set {\rm $\mathcal{R}^{(\text{fg})}$} with FDR control.
	 		}
	 	\end{proof}
	 	
	 	
	 	\section{Detailed Definition of Catching Sets of Resolution-Adaptive Approaches}\label{df_cc}
	 	
	 	Catching sets of resolution-adaptive approaches are formally defined as follows.
	 	
	 	\begin{itemize}
	 		\item[$\star$] For KnockoffZoom (2 layers), based on selection sets $\mathcal{R}_{\text{KZ}}^{(0)},\mathcal{R}_{\text{KZ}}^{(2)}$ at layers 0 and 2,  we derive catching sets as 
	 		$$\begin{aligned}
	 			&\mathcal{CC}_{\text{KZ}}=\mathcal{CC}_{\text{KZ}}^{(0)}\cup \mathcal{CC}_{\text{KZ}}^{(2)},\\
	 			\text{where }&\mathcal{CC}_{\text{KZ}}^{(0)}=
	 			\{\mathcal{S}_{j,\text{KZ}}^{(0)}=\{j\}| j\in \mathcal{R}_{\text{KZ}}^{(0)}\},\\
	 			&\mathcal{CC}_{\text{KZ}}^{(2)}=\{\mathcal{S}_{k,\text{KZ}}^{(2)}={B}^{(2)}_{k}| k\in \mathcal{R}_{\text{KZ}}^{(2)} \text{ and  } j\notin\mathcal{R}_{\text{KZ}}^{(0)}\text{ for all }j\in  {B}^{(2)}_{k}\}.\\
	 		\end{aligned}$$
	 		\item[$\star$] For KnockoffZoom (3 layers), based on selection sets $\mathcal{R}_{\text{KZ}}^{(0)},\mathcal{R}_{\text{KZ}}^{(1)},\mathcal{R}_{\text{KZ}}^{(2)}$ at layers 0, 1 and 2,  we derive catching sets as 
	 		$$\begin{aligned}
	 			&\mathcal{CC}_{\text{KZ}}=\mathcal{CC}_{\text{KZ}}^{(0)}\cup \mathcal{CC}_{\text{KZ}}^{(1)}\cup \mathcal{CC}_{\text{KZ}}^{(2)},\\
	 			\text{where }&\mathcal{CC}_{\text{KZ}}^{(0)}=
	 			\{\mathcal{S}_{j,\text{KZ}}^{(0)}=\{j\}| j\in \mathcal{R}_{\text{KZ}}^{(0)}\},\\
	 			&\mathcal{CC}_{\text{KZ}}^{(1)}=\{\mathcal{S}_{k,\text{KZ}}^{(1)}={B}^{(1)}_{k}| k\in \mathcal{R}_{\text{KZ}}^{(1)} \text{ and  } j\notin\mathcal{R}_{\text{KZ}}^{(0)}\text{ for all }j\in  {B}^{(1)}_{k}\},\\
	 			&\mathcal{CC}_{\text{KZ}}^{(2)}=\{\mathcal{S}_{k,\text{KZ}}^{(2)}={B}^{(2)}_{k}| k\in \mathcal{R}_{\text{KZ}}^{(2)} \text{ and  } k'\notin\mathcal{R}_{\text{KZ}}^{(1)}\text{ for all }{B}^{(1)}_{k'}\subset  {B}^{(1)}_{k}\},\\
	 		\end{aligned}$$
	 		
	 		\item[$\star$] For KeLP (2 layers), based on selection sets $\mathcal{R}_{\text{KeLP}}$, we derive catching sets as 
	 		$$\begin{aligned}
	 			&\mathcal{CC}_{\text{KeLP}}=\mathcal{CC}_{\text{KeLP}}^{(0)}\cup \mathcal{CC}_{\text{KeLP}}^{(2)},\\
	 			\text{where }& \mathcal{CC}_{\text{KeLP}}^{(0)}=\{\mathcal{S}_{j,\text{KeLP}}^{(0)}=\{j\}| \{j\}\in \mathcal{R}_{\text{KeLP}}\},\\
	 			&\mathcal{CC}_{\text{KeLP}}^{(2)}=\{\mathcal{S}_{k,\text{KeLP}}^{(2)}={B}^{(2)}_{k}| {B}^{(2)}_{k}\in \mathcal{R}_{\text{KeLP}}\}.\\
	 		\end{aligned}$$
	 		
	 		\item[$\star$] For KeLP (3 layers), based on selection sets $\mathcal{R}_{\text{KeLP}}$, we derive catching sets as 
	 		$$\begin{aligned}
	 			&\mathcal{CC}_{\text{KeLP}}=\mathcal{CC}_{\text{KeLP}}^{(0)}\cup \mathcal{CC}_{\text{KeLP}}^{(1)}\cup \mathcal{CC}_{\text{KeLP}}^{(2)},\\
	 			\text{where }& \mathcal{CC}_{\text{KeLP}}^{(0)}=\{\mathcal{S}_{j,\text{KeLP}}^{(0)}=\{j\}| \{j\}\in \mathcal{R}_{\text{KeLP}}\},\\
	 			&\mathcal{CC}_{\text{KeLP}}^{(1)}=\{\mathcal{S}_{k,\text{KeLP}}^{(1)}={B}^{(1)}_{k}| {B}^{(1)}_{k}\in \mathcal{R}_{\text{KeLP}}\}.\\
	 			&\mathcal{CC}_{\text{KeLP}}^{(2)}=\{\mathcal{S}_{k,\text{KeLP}}^{(2)}={B}^{(2)}_{k}| {B}^{(2)}_{k}\in \mathcal{R}_{\text{KeLP}}\}.\\
	 		\end{aligned}$$
	 		
	 	\end{itemize}
	 	
	 	\section{Our Implementation of KeLP \citep{Gablenz2024}}\label{KeLP}
	 	
	 	We describe our detailed implementation of KeLP (2 layers) under the simulated experiments in Section \ref{Simulation}, while the implementation of KeLP (3 layers) is analogous.
	 	
	 	As a procedure based on $e$-value, KeLP first compute $e$-values for all $H_j^{\text{(ff)}}$'s at the feature-versus-feature level and all $H_k^{\text{(gg)}}$'s at the group-versus-group level as follows.
	 	
	 	\begin{itemize}
	 		\item For $j=1,\ldots,p$ where $p=1157$, the $e$-value $e_j^{\text{(ff)}}$ of $H_j^{\text{(ff)}}$ is computed as 
	 		$$e_j^{\text{(ff)}}=\frac{c^{(\text{ff})}\times I(W_j^{\text{(ff)}}\geq t^{\text{(ff)}}_{\gamma^{\text{(ff)}}})}{1+\#\{j |W_j^{\text{(ff)}}\leq -t^{\text{(ff)}}_{\gamma^{\text{(ff)}}}\}},$$
	 		where $W_j^{\text{(ff)}}$ is the feature importance score of $H_j^{\text{(ff)}}$ computed using model-X knockoff and $$	t^{\text{(ff)}}_{\gamma^{\text{(ff)}}}=\min\Biggl\{t>0\Bigg|\frac{1+\#\{j |W_j^{\text{(ff)}}\leq -t\}}{1\vee\#\{j |W_j^{\text{(ff)}}\geq t\}}\leq \gamma^{\text{(ff)}}\Biggr\}.$$
	 		In other words, $e_j^{\text{(ff)}}$'s are $e$-values computed by model-X knockoff  filter under the target FDR level $\gamma^{\text{(ff)}}$.
	 		
	 		\item The $e$-value $e_k^{\text{(gg)}}$ of $H_k^{\text{(gg)}}$ with respect to ${B}^{(2)}_{k}$ is computed as 
	 		$$e_k^{\text{(gg)}}=\frac{c^{(\text{gg})}\times I(W_k^{\text{(gg)}}\geq t^{\text{(gg)}}_{\gamma^{\text{(gg)}}})}{1+\#\{j |W_k^{\text{(gg)}}\leq -t^{\text{(gg)}}_{\gamma^{\text{(gg)}}}\}},$$
	 		where $W_k^{\text{(gg)}}$ is the feature importance score of $H_k^{\text{(gg)}}$ computed using group knockoff with respect to ${B}^{(2)}_{1},\ldots,{B}^{(2)}_{345}$ and $$	t^{\text{(gg)}}_{\gamma^{\text{(gg)}}}=\min\Biggl\{t>0\Bigg|\frac{1+\#\{j |W_k^{\text{(gg)}}\leq -t\}}{1\vee\#\{j |W_k^{\text{(gg)}}\geq t\}}\leq \gamma^{\text{(gg)}}\Biggr\}.$$
	 		In other words, $e_k^{\text{(gg)}}$'s are $e$-values computed by group knockoff filter with respect to ${B}^{(2)}_{1},\ldots,{B}^{(2)}_{345}$ under the target FDR level $\gamma^{\text{(gg)}}$.
	 	\end{itemize}
	 	
	 	Based on $e_j^{\text{(ff)}}$'s and $e_k^{\text{(gg)}}$'s, KeLP returns the resolution-adaptive selection set under the FDR level $\alpha$ via solving the constrained optimization problem
	 	\begin{equation}
	 		\label{opt2}
	 		\begin{aligned}[b]
	 			&\max\left(\sum_{j=1}^{1157}x_j^{(\text{ff})}+\sum_{k=1}^{345}\frac{x_k^{(\text{gg})}}{|{B}^{(2)}_{k}|}\right),\\
	 			\text{s.t. }&x_1^{(\text{ff})},\ldots,x_{1157}^{(\text{ff})},x_1^{(\text{gg})},\ldots,x_{345}^{(\text{gg})}\in\{0,1\},\\
	 			&D_{\text{total}}-\alpha e_j^{(\text{ff})}\left(\sum_{j=1}^{1157}x_j^{(\text{ff})}+\sum_{k=1}^{345}{x_k^{(\text{gg})}}\right)\leq 	D_{\text{total}}\times (1-x_j^{(\text{ff})}),\quad j=1,\ldots,1157,\\
	 			&D_{\text{total}}-\alpha e_k^{(\text{gg})}\left(\sum_{j=1}^{1157}x_j^{(\text{ff})}+\sum_{k=1}^{345}{x_k^{(\text{gg})}}\right)\leq 	D_{\text{total}}\times (1-x_k^{(\text{gg})}),\quad k=1,\ldots,345,\\
	 			&x_j^{(\text{ff})}+\sum_{k=1}^{345}x_k^{(\text{gg})}\times \text{I}(j\in {B}^{(2)}_{k})\leq 1,\quad j=1,\ldots,1157.
	 		\end{aligned}
	 	\end{equation}
	 	Here, $D_{\text{total}}=1157+345$ equals the total number of $e$-values computed at both levels. Finally, catching sets can be obtained as $\{j|\hat{x}_j^{(\text{ff})}=1\}\cup\{{B}^{(2)}_{k}|\hat{x}_k^{(\text{gg})}=1\}$ based on the solution $(\hat{x}_1^{(\text{ff})},\ldots,\hat{x}_{1157}^{(\text{ff})},\hat{x}_1^{(\text{gg})},\ldots,\hat{x}_{345}^{(\text{gg})})$ of \eqref{opt2}.
	 	
	 	In \citet{Gablenz2024}, it is said that choice of $c^{(\text{ff})}$, $c^{(\text{gg})}$ and $\gamma^{(\text{ff})}$, $\gamma^{(\text{gg})}$  are left to users. Thus, in our simulation (Section \ref{Simulation}), we use $c^{(\text{ff})}=1157$ and $c^{(\text{gg})}=345$ to equal the number of hypotheses in both levels. Suppose $\mathcal{R}^{\text{(ff)}}_{\gamma^{\text{(ff)}}}$ is the selection set computed by model-X knockoff  filter under the target FDR level $\gamma^{\text{(ff)}}$ and $\mathcal{R}^{\text{(gg)}}_{\gamma^{\text{(gg)}}}$ is the selection set computed by the group knockoff  filter under the target FDR level $\gamma^{\text{(gg)}}$, we have that 
	 	\begin{itemize}
	 		\item $e_j^{\text{(ff)}}\geq \frac{1157}{\gamma^{\text{(ff)}}\times|\mathcal{R}^{\text{(ff)}}_{\gamma^{\text{(ff)}}}|}$
	 		for all $j \in \mathcal{R}^{\text{(ff)}}_{\gamma^{\text{(ff)}}}$;
	 		\item $e_k^{\text{(gg)}}\geq \frac{345}{\gamma^{\text{(gg)}}\times|\mathcal{R}^{\text{(gg)}}_{\gamma^{\text{(gg)}}}|}$
	 		for all $k \in \mathcal{R}^{\text{(gg)}}_{\gamma^{\text{(gg)}}}$.
	 	\end{itemize}
	 	To avoid the case that no catching set is obtained with nonempty $\mathcal{R}^{\text{(ff)}}_{\gamma^{\text{(ff)}}}$ and $\mathcal{R}^{\text{(gg)}}_{\gamma^{\text{(gg)}}}$, we elaborately tune $\gamma^{(\text{ff})}$, $\gamma^{(\text{gg})}$ as follows. Specifically, our target is to include two trivial points in the feasible region of \eqref{opt2}, including
	 	\begin{itemize}
	 		\item the point where $x_j^{(\text{ff})}=\text{I}(j\in \mathcal{R}^{\text{(ff)}}_{\gamma^{\text{(ff)}}})$ for $j=1,\ldots,1157$ and $x_k^{(\text{gg})}=0$ for $k=1,\ldots,345$;
	 		\item the point where $x_j^{(\text{ff})}=0$ for $j=1,\ldots,1157$ and $x_k^{(\text{gg})}=\text{I}(k\in \mathcal{R}^{\text{(gg)}}_{\gamma^{\text{(gg)}}})$ for $k=1,\ldots,345$.
	 	\end{itemize}  
	 	By doing so, we have nonempty catching sets as long as either $\mathcal{R}^{\text{(ff)}}_{\gamma^{\text{(ff)}}}$ or $\mathcal{R}^{\text{(gg)}}_{\gamma^{\text{(gg)}}}$ is nonempty.
	 	
	 	To achieve this, we require that
	 	$$\begin{aligned}
	 		\frac{1157}{\gamma^{\text{(ff)}}\times|\mathcal{R}^{\text{(ff)}}_{\gamma^{\text{(ff)}}}|}&\geq \frac{D_{\text{total}}}{\alpha \times|\mathcal{R}^{\text{(ff)}}_{\gamma^{\text{(ff)}}}|},\\
	 		\frac{345}{\gamma^{\text{(gg)}}\times|\mathcal{R}^{\text{(gg)}}_{\gamma^{\text{(gg)}}}|}&\geq \frac{D_{\text{total}}}{\alpha \times|\mathcal{R}^{\text{(gg)}}_{\gamma^{\text{(gg)}}}|},
	 	\end{aligned}$$
	 	from the third and forth line of \eqref{opt2}. This leads to our choice that $\gamma^{\text{(ff)}}=1157/D_{\text{total}}\times \alpha$ and $\gamma^{\text{(gg)}}=345/D_{\text{total}}\times \alpha$. When more than $2$ layers are included in the inference, we analogously implement KeLP with $\gamma$ of each layers as $\alpha$ times ratio between the number of hypotheses in this layer and the total number of hypotheses in all layers and the constant $c$ of each layer as the  number of hypotheses in this layer.

	 \end{appendix}

	\bibliographystyle{apalike}
	\bibliography{sample}
	
\end{document}